\documentclass[screen]{acmart}

\usepackage{algorithm}
\usepackage{algorithmic}
\usepackage{url}
\usepackage{microtype}
\usepackage{graphicx}
\usepackage{subcaption}
\usepackage{multirow}
\usepackage{mathpartir}
\usepackage{multicol}

\setcopyright{none}

\usepackage{ulem} 
\usepackage{xcolor}

\usepackage{listings}
\usepackage{color}

\begin{document}

\title{Certified Compilation based on G\"{o}del Numbers}

\author{Guilherme de Oliveira Silva}
  \email{gsilva@cadence.com}
  \orcid{0009-0007-6742-1829}
\affiliation{
  \institution{Cadence}
  \city{Belo Horizonte}
  \country{Brazil}
}

\author{Fernando Magno Quint\~{a}o Pereira}
  \email{fernando@dcc.ufmg.br}
  \orcid{0000-0002-0375-1657}
\affiliation{
  \institution{UFMG}             
  \country{Brazil}
}
\email{fernando@dcc.ufmg.br}          

\begin{abstract}
In his 1984 Turing Award lecture, Ken Thompson showed that a compiler could be maliciously altered to insert backdoors into programs it compiles and perpetuate this behavior by modifying any compiler it subsequently builds. Thompson's hack has been reproduced in real-world systems for demonstration purposes.
Several countermeasures have been proposed to defend against Thompson-style backdoors, including the well-known {\it Diverse Double-Compiling} (DDC) technique, as well as methods like translation validation and CompCert-style compilation. However, these approaches ultimately circle back to the fundamental question: ``{\it How can we trust the compiler used to compile the tools we rely on?}''
In this paper, we introduce a novel approach to generating certificates to guarantee that a binary image faithfully represents the source code. These certificates ensure that the binary contains all and only the statements from the source code, preserves their order, and maintains equivalent def-use dependencies. The certificate is represented as an integer derivable from both the source code and the binary using a concise set of derivation rules, each applied in constant time.
To demonstrate the practicality of our method, we present \textsc{Charon}, a compiler designed to handle a subset of C expressive enough to compile \textsc{FaCT}, the Flexible and Constant Time cryptographic programming language.
\end{abstract}




\settopmatter{printacmref=false} 
\renewcommand\footnotetextcopyrightpermission[1]{} 
\pagestyle{plain} 

\maketitle

\section{Introduction}
\label{sec_introduction}

%
In his Turing Award lecture, ``{\it Reflections on Trusting Trust''}~\cite{Thompson84}, Ken Thompson introduced an influential concept in cybersecurity: he described an attack in which a compiler is modified to insert a backdoor into the UNIX login command. The brilliance of this attack lies in its ability to self-propagate: the compromised compiler not only injects the backdoor but also recognizes and reinserts it into new compiler versions, even if the malicious code is removed from the source.
The concrete implementation of the attach is relatively simple~\cite{Cox23}, something noticeable, given how hard is to prevent it.
Thompson’s work highlights two critical truths: software tools can be compromised in ways that evade traditional inspection, and trust must extend beyond source code to include the entire toolchain. By exposing this hidden vulnerability, Thompson laid the groundwork for understanding and addressing {\it software supply chain attacks}~\cite{Cesarano24, Ladisa23, Ohm20, Okafor24}, a concern that remains relevant today.

\paragraph{The Challenge of Building Trust}
In Thompson's word, instead of trusting that a program is free of Trojan horses, ``{\it Perhaps it is more important to trust the people who wrote the software}''~\cite{Thompson84}.
Nevertheless, various techniques have been developed to enhance trust in software and mitigate the risks of Thompson-style backdoors. One prominent approach specifically addressing this issue is the Diverse Double-Compiling technique, proposed by David A. Wheeler in his PhD dissertation~\cite{Wheeler05,Wheeler10}.
As Section~\ref{sub_ddc} further explains, this method compares binaries produced by different compilers to identify malicious behavior in a potentially compromised compiler.

Going beyond Wheeler's certification methodology, programming language techniques provide broader mechanisms to build trust in software. Translation validation~\cite{Pnueli98, Necula00} verifies that the output of a compiler faithfully preserves the semantics of the input program by validating each individual compilation rather than the entire compiler. Proof-carrying code (PCC)~\cite{Necula97} allows a program consumer to ensure that code provided by an untrusted producer satisfies specific safety properties without needing to trust the producer.
Ultimately, compilers can be fully verified. For instance, \textsc{CompCert}~\cite{Leroy09} is a compiler whose correctness is formally guaranteed via the Coq proof assistant.
Together, these techniques strengthen the foundations of trust in the software ecosystem.
And yet, the fundamental issue of ``{\it who we trust}'' still plagues them all.

The aforementioned techniques -- certified compilation (as in CompCert), Proof-Carrying Code (PCC), Translation Validation, and Diverse Double-Compiling (DDC) -- each have limitations in addressing the fundamental challenge raised by Thompson's ``Trusting Trust'' attack.
These limitations arise primarily because demonstrating that a binary faithfully implements the semantics of arbitrary source code requires solving the problem of semantic equivalence, which is undecidable in general due to its relationship to the Halting Problem~\cite{Church36,Turing36}. 

\begin{description}
\item [Proof-Carrying Code:] PCC is designed to demonstrate specific safety properties, such as the absence of out-of-bounds memory accesses (e.g., ensuring all array accesses $\mathtt{arr[i]}$ satisfy $0 \leq \mathtt{i} < \mathtt{size}$). While effective at verifying certain well-defined properties, PCC does not address the presence of malicious behavior, such as Trojan horses, in the binary.

\item [Translation Validation:] The validator verifies the correctness of individual compiler transformations, such as register allocation~\cite{Rideau10} or constant propagation~\cite{Pnueli99}, rather than the entire compiler. By focusing on specific transformations, translation validation avoids undecidability.
However, it is inherently limited to the scope of the program and transformations being validated. If validation fails, it may indicate a limitation in the validation technique rather than an actual error in the compiler.

\item [Verified Compilation:] Unlike PCC, a verified compiler such as \textsc{CompCert} does not produce standalone certificates that users can independently verify. Instead, \textsc{CompCert} provides strong semantic correctness guarantees for the programs it compiles, but it does not address the problem of trusting the compiler binary itself. If the binary of the verified compiler is compromised, all guarantees it provides are invalid.

\item [Diverse Double-Compiling:] As Section~\ref{sub_ddc} explains, DDC relies on the assumption that at least one trusted compiler exists. While this technique reduces the likelihood of an undetected backdoor, it cannot eliminate the possibility entirely. If both compilers used in the DDC process are compromised, then Thompson's backdoor may still persist.
\end{description}

Therefore, each of these techniques represents a step toward building trust in software, but none can completely eliminate the risks highlighted by Thompson's seminal work.

\paragraph{The Contribution of This Work:}
This paper proposes a technique to certify that a target (low-level) program was produced as the result of translating a source (high-level) program.  
Said technique is based on the famous encoding Kurt Gödel used in the proof of his incompleteness theorems \cite{Godel31}.  
As explained in Section~\ref{sec_solution}, the certificate of faithful compilation is embedded within the structure of the high-level program and in the structure of the low-level program, rather than existing as a separate artifact, as is the case with proof-carrying code.
Figure~\ref{fig_certificationOverview} illustrates the proposed methodology.

\begin{figure}[ht]
\centering
\includegraphics[width=1\textwidth]{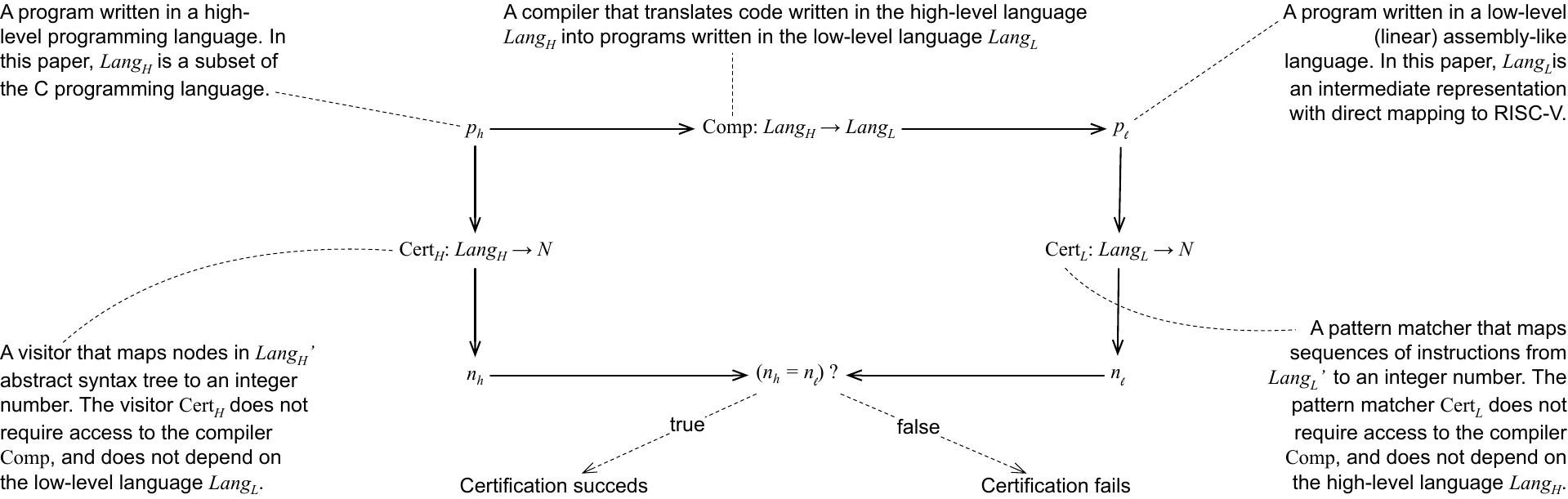}
\caption{Overview of the proposed certification methodology.}
\Description{Overview of the proposed certification methodology.}
\label{fig_certificationOverview}
\end{figure}

In Figure~\ref{fig_certificationOverview}, \(p_h\) is a program written in a high-level source language \(\mathit{Lang}_H\), and \(p_\ell\) is the program written in a low-level target language \(\mathit{Lang}_L\), produced by a compiler \(\mathit{Comp}\), i.e., \(p_\ell = \mathit{Comp}(p_h)\).
Figure~\ref{fig_certificationOverview} defines two mapping algorithms, \(H:\mathit{Lang}_H \mapsto \mathbb{N}\) and \(L:\mathit{Lang}_L \mapsto \mathbb{N}\).
If $\mathit{Cert}_H(p_h) = \mathit{Cert}_L(\mathit{Comp}(p_h))$, then $\mathit{Comp}$ has not modified the compilation contract, meaning that $\mathit{Comp}$ has generated only the instructions necessary to implement in $p_\ell$ the semantics of $p_h$.
Example~\ref{ex_exampleCertificate} will clarify these ideas.

\begin{example}
\label{ex_exampleCertificate}
Figure~\ref{fig_exampleCertificate} shows how we produce a certificate for a piece of code that is often considered the shortest valid C program.
This example illustrates some basic principles behind our approach.
First, a certificate is a product of powers of prime numbers.
The base of such exponents, e.g., numbers such as 2, 3 and 5 in Figure~\ref{fig_exampleCertificate} mark the position where each syntactic construct appears (first, second, third, etc) in the program's AST, or in the low-level intermediate representation (IR) that we adopt in this paper.
The exponents are associated with particular programming constructs in the high-level language, or with particular sequences of instructions in the low-level language.
For instance, a $\mathtt{return}$ statement in the high-level language is associated with the exponent 19.
The code that corresponds to a $\mathtt{return}$ statement in the low-level language uses two opcodes: $\mathtt{MOV}$ and $\mathtt{JR}$.
The certificate of a program is the product of the certificate of all the constructs that make up that program.
This sequence is also associated with exponent 19.
In this example, this certificate is $2^2 \times 3^{19} \times 5^{113}$.
In this paper we adopt as the low-level language a linear intermediate representation (IR) that has one-to-one correspondents with RISC-V instructions, as Figure~\ref{fig_exampleCertificate} shows.
Thus, by certifying the intermediate representation, we also certify its corresponding RISC-V image.
\end{example}

\begin{figure}[ht]
\centering
\includegraphics[width=1\textwidth]{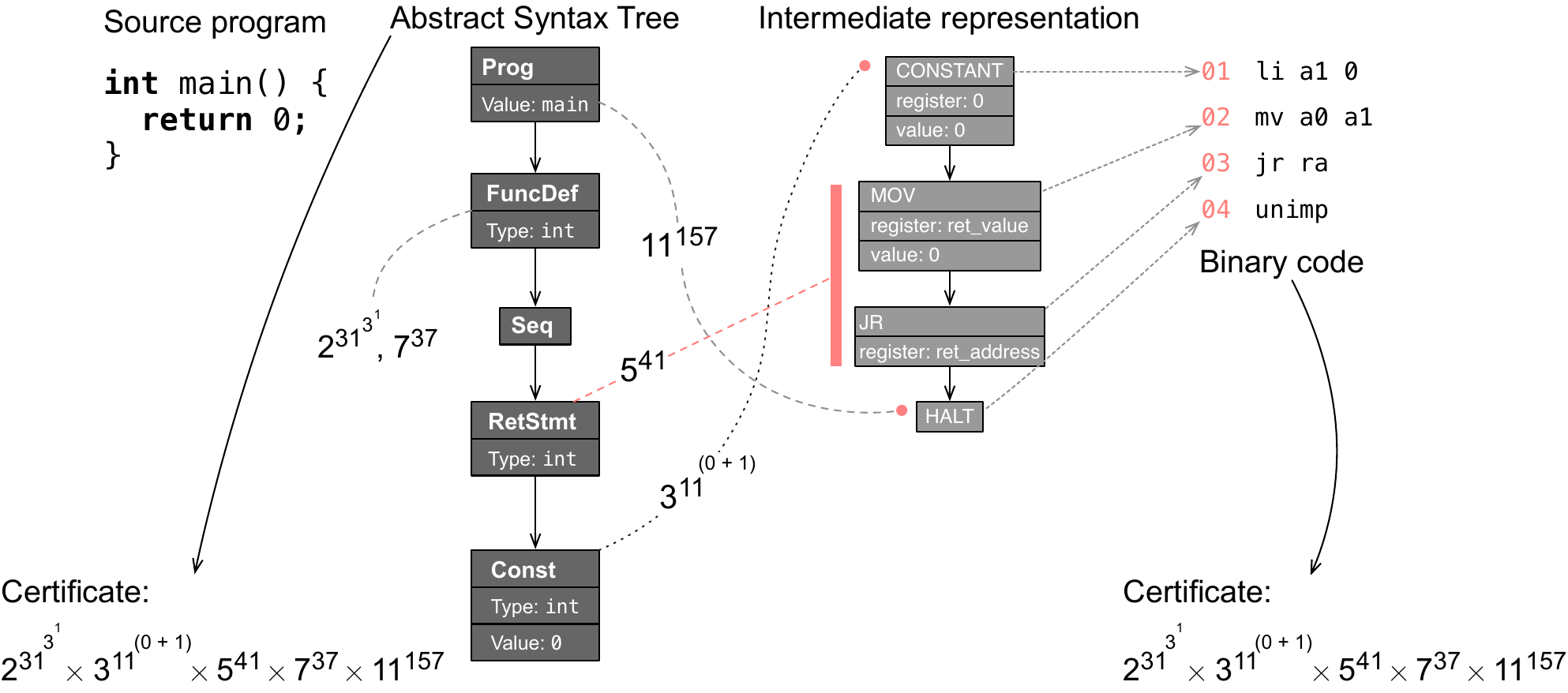}
\caption{How \textsc{Charon} generates certificates for a simple C program.}
\Description{How \textsc{Charon} generates certificates for a simple C program.}
\label{fig_exampleCertificate}
\end{figure}

Example~\ref{ex_exampleCertificate} illustrates how the proposed methodology works on a very simple program.
This example misses two capacities of this approach: the ability to certify high-level constructs that are translated to non-contiguous sequences of instructions (like loops and branches); and the ability to encode def-use relations, as we would have once the high-level program manipulates variables.
Section~\ref{sec_solution} will explain how we deal with such constructs.

\paragraph{Summary of Results}
We have implemented the ideas proposed in this paper in a compilation framework called \textsc{Charon}.
This compiler is capable of certifying a subset of C that is expressive enough to support \textsc{FaCT}~\cite{Cauligi19} programs.
\textsc{FaCT} is a restricted version of C designed to facilitate constant-time programming. It provides users with three control-flow constructs: \texttt{if-then-else}, \texttt{while}, and \texttt{return}, along with a type system that allows for the specification of public and classified data.
\texttt{Charon} supports all the control-flow constructs in \textsc{FaCT}, three data types (\texttt{int16}, \texttt{int32}, and \texttt{float}), and implicit type coercions.
Henceforth, we call this version of \textsc{FaCT} (the $\mathit{Lang}_H$ in Figure~\ref{fig_certificationOverview}), \textsc{CharonLang}.
Similarly, we name \textsc{CharonIR} our target low-level intermediate representation (the $\mathit{Lang}_L$ in Figure~\ref{fig_certificationOverview}).

Our certification methodology provides the guarantees shown in Figure~\ref{fig_certificationOverview}, along with the property of \emph{invertibility}.
Invertibility means that there exists a function $\mathit{Canon} : \mathbb{N} \mapsto \mathit{Lang}_H$ such that, if $\mathit{Cert}_H(P_H) = n$, then $\mathit{Canon}(n) = P_H'$, where $P_H'$ is the \emph{Canonical Representation} of program $P_H$.
As explained in Section~\ref{sub_canonical}, $P_H$ and $P_H'$ need not be syntactically identical, but they are guaranteed to be semantically equivalent.
This result mimic G\"{o}del's usage of the {\it Fundamental Theorem of Arithmetic}, which states that every integer greater than 1 can be expressed uniquely (up to the order of factors) as a product of prime numbers.
Building on this fact, Section~\ref{sub_correctness} shows that certificates are invertible at both the high-level and low-level representations of programs.

The absolute value of a certificate, e.g., measured as the number of bits necessary to represent it, grows exponentially with the size of programs.
This asymptotic behavior characterizes $\mathit{Cert}_H$ and $\mathit{Cert}_L$.
In the former case, because a certificate contains a multiplicative factor for each construct in the program's AST.
In the latter, because the certificate contains a multiplicative factor for each instruction in the program's low-level representation.
Nevertheless, as we show in Section~\ref{sec_eval}, the symbolic representation of a certificate, as a string that encodes a series of multiplications, is still linear on the size of the program, be it written in \textsc{CharonLang} or in \textsc{CharonIR}.

\section{Trusting Trust}
\label{sec_overview}

This section provides the background knowledge necessary to follow the ideas that will be developed in Section~\ref{sec_solution}.
To this end, Section~\ref{sub_thompson} starts explaining how Thompson's hack works.
Section~\ref{sub_godel}, in turn, discusses the principles behind G\"{o}del's Numbering System.

\subsection{The Thompson Hack}
\label{sub_thompson}

Ken Thompson's famous ``Trusting Trust'' hack demonstrates a self-replicating backdoor embedded into a compiler -- the ``Trojan Compiler''.
Thompson's exploit starts with a clean compiler, which does not contain any backdoors. This compiler is trusted by the user and produces backdoor-free binaries.
The attack leverages two backdoors to propagate malicious behavior, even when the source code appears clean. The process can be described in four stages, as illustrated in Figure~\ref{fig_thompsonHack}.

\begin{figure}[ht]
\centering
\includegraphics[width=0.9\textwidth]{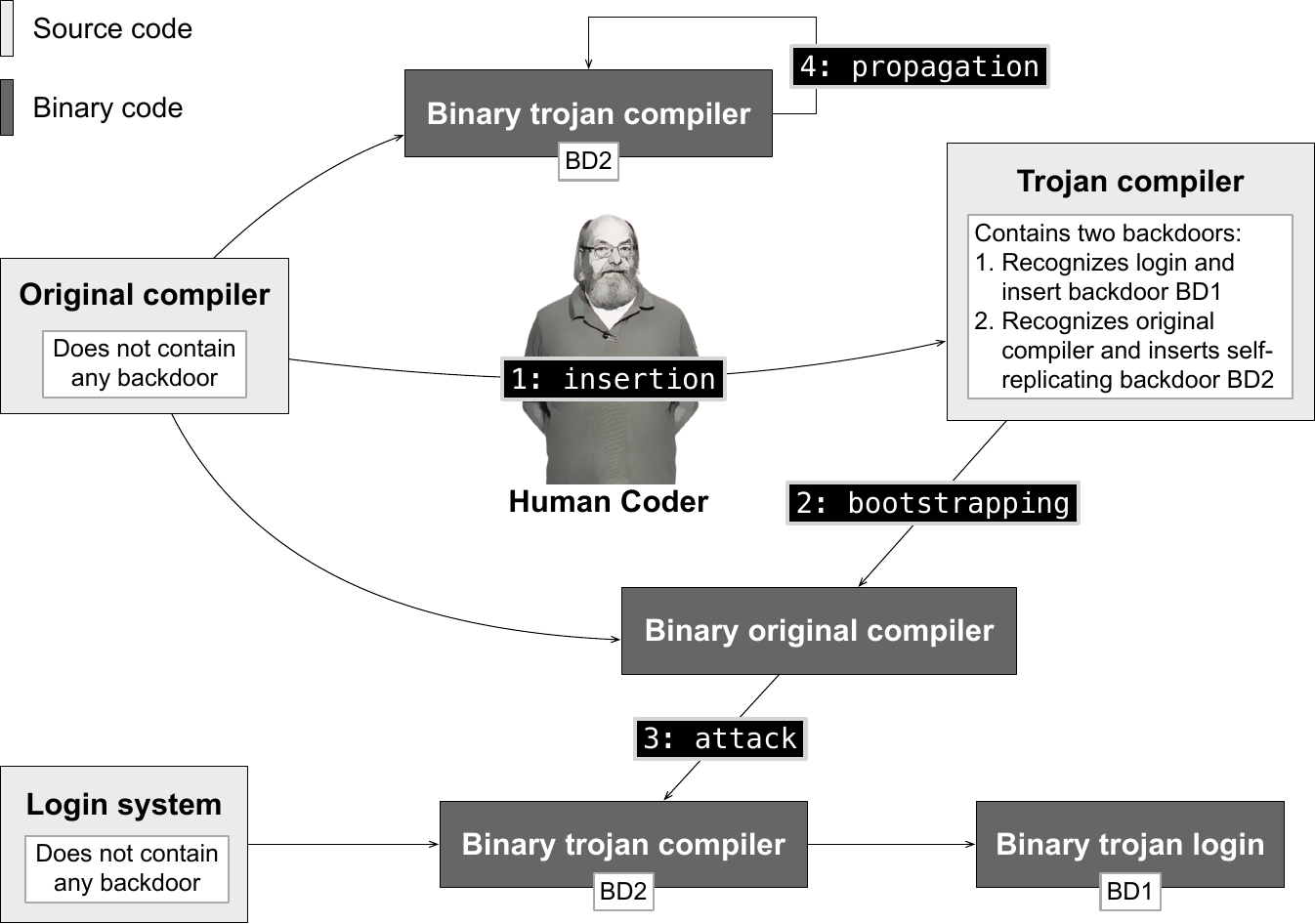}
\caption{A schematic view of Thompson's ``Trusting Trust'' attack.}
\Description{A schematic view of Thompson's ``Trusting Trust'' attack.}
\label{fig_thompsonHack}
\end{figure}

\begin{enumerate}
\item \textbf{Insertion:}
In Thompson's hack, a (malicious) user modifies the compiler source to produce a ``trojan compiler'', which contains two backdoors:
\begin{itemize}
\item \textbf{BD1:} This backdoor recognizes when the trojan compiler is compiling the \texttt{login} program and inserts a malicious backdoor (e.g., a hardcoded password) into the compiled binary.
\item \textbf{BD2:} This self-replicating backdoor recognizes when the trojan compiler is compiling its own source code and reinserts both BD1 and BD2 into the new compiler binary.
\end{itemize}

\item \textbf{Bootstrapping:}
Once the trojan compiler's source is compiled with the original compiler, it produces a binary version of the trojan compiler. This binary now contains both backdoors, BD1 and BD2.
At this point, the source code of the trojan compiler is no longer necessary, and can be removed.
Therefore, the source code of the compiler appears clean, as BD2 exists only at the binary level.

\item \textbf{Attack:}
When the binary trojan compiler compiles the \texttt{login} system, it inserts BD1 into the resulting binary, creating a \texttt{login} login system.

\item \textbf{Propagation:}
When the binary trojan compiler compiles a clean version of its own source code, it reinserts BD1 and BD2 into the new binary compiler. This ensures the trojan behavior propagates indefinitely, even if the source code is reviewed and found to be clean.
\end{enumerate}

\noindent
The genius of this attack lies in its self-propagating nature: the malicious behavior becomes embedded at the binary level and is undetectable through inspection of the source code alone. Trusting the compiler binary is therefore essential, as the hack demonstrates how source code auditing alone cannot guarantee the absence of malicious behavior.

\subsection{G\"{o}del Numbering}
\label{sub_godel}

We aim to certificate programs by assigning unique identifiers to each instruction that encodes the operation, its operands, and its relative position in the program simultaneously. These identifiers are then combined to produce a single number that uniquely represents a program in such manner that any changes made to the order of operations or operands will produce a different certificate. This is accomplished with a numbering system based on {\it G\"{o}del Numbers}.
This numbering system works according to the following sequence of steps, which Example~\ref{ex_godel} will illustrate:

\begin{enumerate}
\item \textbf{Assign Unique Numbers to Symbols:}  
Create an \textit{encoding table} where each symbol in the logical system (e.g., variables, operators, parentheses) is assigned a unique integer.

\item \textbf{Map Each Position to a Prime Number:}  
Use the sequence of prime numbers (\(2, 3, 5, 7, \dots\)) to represent the positions of symbols in the expression.  
The \(n\)-th symbol in the expression corresponds to the \(n\)-th prime number.

\item \textbf{Calculate the Contribution of Each Symbol:}  
For each symbol \(s_i\) at position \(i\):  
\begin{itemize}
  \item Take the \(i\)-th prime \(p_i\) as the \textit{base}.
  \item Use the integer assigned to \(s_i\) from the encoding table as the \textit{exponent}.
  \item Compute \(p_i^{e_i}\), where \(e_i\) is the encoded value of \(s_i\).
\end{itemize}

\item \textbf{Multiply All Contributions:}  
Compute the Gödel number by multiplying the contributions of all symbols:  
\[
G = \prod_{i=1}^{n} p_i^{e_i}
\]
where \(n\) is the total number of symbols in the expression.
\end{enumerate}

\begin{example}
\label{ex_godel}
Figure~\ref{fig_godelExample} shows how we can encode the expression $a + b$ using a G\"{o}del Numbering System.
Let us assume that our example system assigns the following numbers to each potential symbol in our arithmetic system:
$a \mapsto 4, + \mapsto 3, b \mapsto 7$.
Similarly, let us assume that the order in which symbols appear in the logic expression are associated with the primes in ascending order, e.g.: $1^{st} \mapsto 2, 2^{nd} \mapsto 3, 3^{rd} \mapsto 5, \ldots$.
Given these assumptions, we compute the contribution of each one of the three symbols in the expression as follows:
$2^4, 3^3$ and $5^7$.
The final encoding is the product of all these parcels, e.g., $16 \times 27 \times 78,125 = 33,750,000$.
\end{example}

\begin{figure}[ht]
\centering
\includegraphics[width=0.9\textwidth]{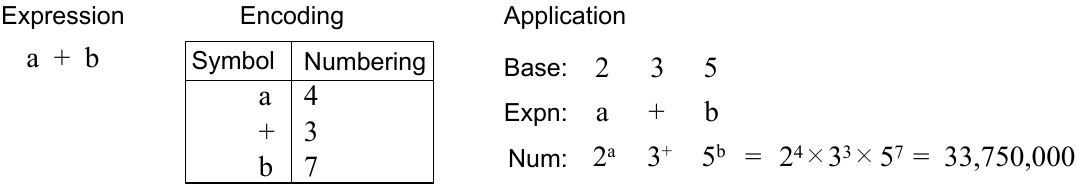}
\caption{An example of encoding based on G\"{o}del Numbers.}
\Description{An example of encoding based on G\"{o}del Numbers.}
\label{fig_godelExample}
\end{figure}

G\"{o}del's numbering relies on the Fundamental Theorem of Arithmetic~\cite{gauss1966disquisitiones}, which states that every integer has a unique prime factorization.
This ensures that each G\"{o}del number maps uniquely to one expression and that the original expression can be reconstructed by factorizing the number.
Thus, the G\"{o}del Numbering System achieves two properties: {\it uniqueness} and {\it reversibility}.
The numbering system proposed in this work also meets these two properties.

\section{Program Certification}
\label{sec_solution}

Our solution to the \textit{trusting trust} problem is implemented by the \textsc{Charon} compiler.
It includes a toy language, \textsc{CharonLang}, that covers a large subset of the C language, and a compiler that targets an intermediate representation (\textsc{CharonIR}) inspired by the RISC-V architecture \cite{Asanovic14}.

\paragraph{\textsc{CharonLang}}
\label{par_toy_language}
The \textsc{Charon} toy language supports all the control-flow constructs present in \textsc{FaCT}, the Flexible and Constant Time cryptographic programming language \cite{Cauligi19}, namely forward and backward branches (loops) plus non-recursive function calls.
In this regard, we adopt the implementation of \textsc{FaCT} available in the work of \citet{Soares23}.
Additionally, this subset of the C language supports integers (16-bit $\mathtt{short}$, 32-bit $\mathtt{int}$) and 32-bit floating point (\texttt{float}) types, plus (implicit) type casts\footnote{The compiler emits adequate type cast instructions to compatibilize operands. Users can not explicitly make type coercion operations.} and all the ANSI C unary and binary mathematical, logical, and bit-wise operations.
\textsc{CharonLang} also supports static arrays and user-defined structures that combine built-in types.
Similarly to \textsc{FaCT}, we chose not to add support to pointer arithmetic to \textsc{CharonLang}.
Figure~\ref{fig_charonGrammar} presents its grammar specification.

\begin{figure}[tbh]
\centering
\begin{footnotesize}
\input{assets/charonGrammar.tex}
\end{footnotesize}
\caption{The \textsc{Charon} Toy Language grammar.}
\Description{The \textsc{Charon} Toy Language grammar.}
\label{fig_charonGrammar}
\end{figure}

\begin{example}
\label{ex_charlang}
Figure~\ref{fig_charonLang} (a) shows the \textsc{CharonLang} implementation of a program that computes the maximum common divisor of two integers.
Part (b) of that figure shows the program's abstract syntax tree.
The compiler ($\mathit{Comp}$ in Fig.~\ref{fig_certificationOverview}) visits each node of the AST to produce the low-level program representation seen in Figure~\ref{fig_charonLang} (c).
Similarly, the high-level certification algorithm ($\mathit{Cert}_H$ in Fig.~\ref{fig_certificationOverview}) visits the AST to produce the program's certificate, as Section~\ref{sub_source_code} will explain.
\end{example}

\begin{figure}
\centering
\includegraphics[width=1\linewidth]{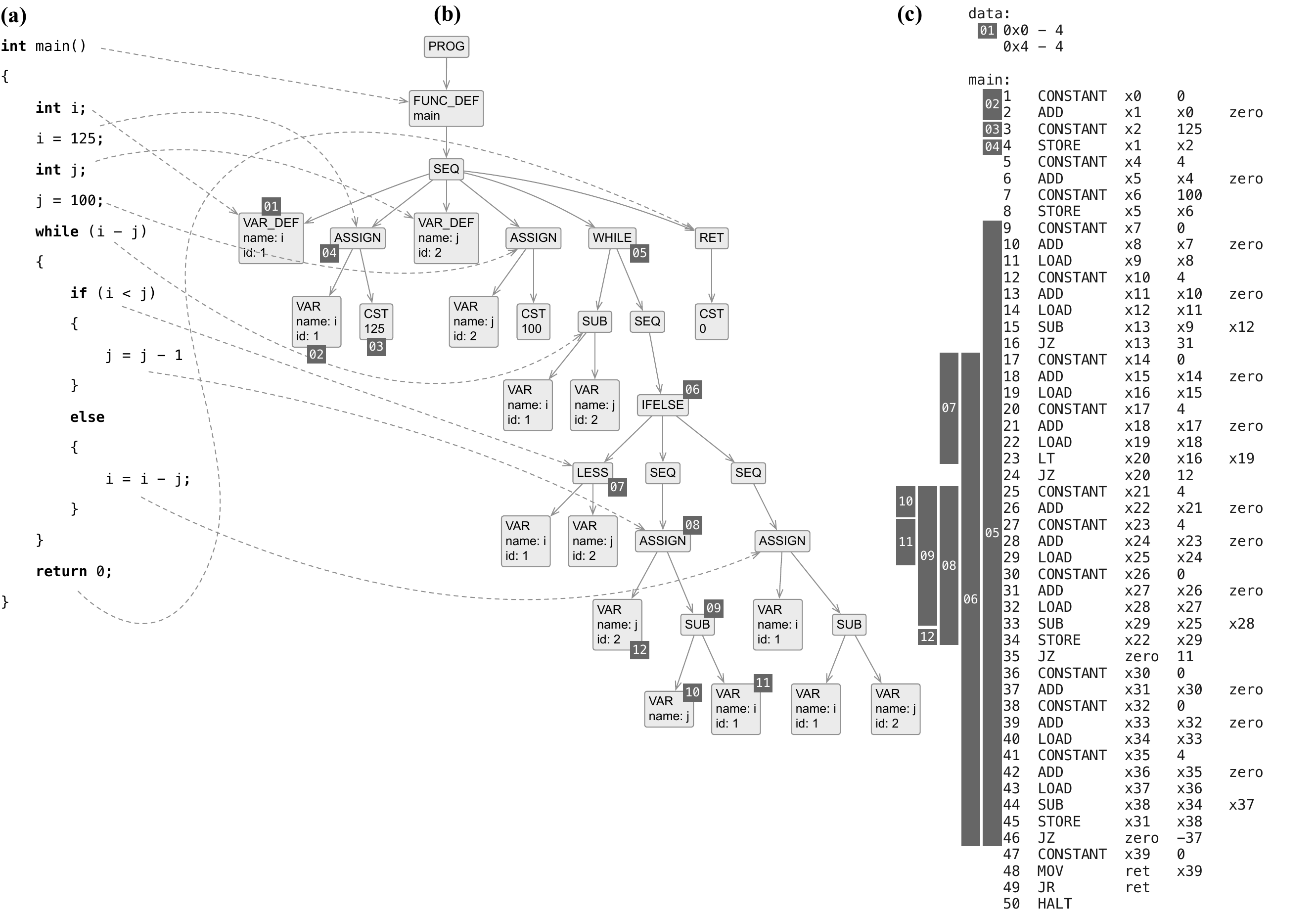}
\caption{(a) Maximum common divisor implemented in \textsc{CharonLang}
(b) Abstract syntax tree.
(c) Low-level representation of the algorithm written in \textsc{CharonIR}.}
\Description{The \textsc{CharonLang} and the \textsc{CharonIR}.}
\label{fig_charonLang}
\end{figure}

\paragraph{Intermediate Representation}
The \textsc{Charon} Intermediate Representation, \textsc{CharonIR}, uses the set of instructions presented in Figure ~\ref{fig_instructionSet}.
Notice that the IR features floating-point variations for all the operations, except the bit-wise ones.
These variations have been omitted for brevity.
Most instruction parameters are registers, except for \texttt{imm}, which takes a constant literal.

\begin{figure}[tbh]
\centering
\begin{footnotesize}
\input{assets/instructionSet.tex}
\end{footnotesize}
\caption{The \textsc{Charon} intermediate representation.
We let \textsc{Rd} be destination and \textsc{R1/R2} be source registers.}
\Description{The \textsc{Charon} Virtual Machine instruction set.}
\label{fig_instructionSet}
\end{figure}

\begin{example}
\label{ex_charonir}
Figure~\ref{fig_charonLang} (c) shows the three-address code representation of the implementation of maximum common divisor seen in
Figure~\ref{fig_charonLang} (a).
This example highlights that any value, variable or constant must be loaded into a temporary register before use.
\textsc{CharonIr} provides the \texttt{CONSTANT} instruction to load constants, because this explicit handling of literals simplifies the construction of their certificates, as Section~\ref{sub_machine_code} will explain.
Retrieving the contents of a variable with \texttt{LOAD} requires the compiler to save its address in a temporary register first.
This is achieved in two steps: first, it will emit a \texttt{CONSTANT} with its temporary address; then, it will compute the offset, if any, and store the final address in the register to load from.
Following RISC-V, an object file is exported with minimal metadata; namely, a \texttt{data} section that describes the lengths of variables identified by their base addresses.
\end{example}

As Example~\ref{ex_charonir} illustrates, \textsc{Charon}'s intermediate representation follows the RISC-V syntax.
The primary difference from RISC-V is that \textsc{CharonIR} has access to an infinite number of register names.
Thus, \textsc{CharonIR} registers should be understood as memory locations instead of typical registers.
Despite this, the IR includes some predefined registers, adhering to the typical RISC-V/Linux Application Binary Interface.
These predefined registers include one for storing the return address of functions, others for holding function arguments and return values, and a dedicated register that always contains zero for convenience.  
Variables in \textsc{CharonIR} can be defined at multiple allocation sites, meaning the IR does not conform to the static single-assignment (SSA) format~\cite{Cytron91}, so pervasive in modern compiler design.
 
\paragraph{Translation Specification}
Intermediate Representation code is generated by parsing \textsc{CharonLang} following the translation rules in Figure~\ref{fig_translationSpecification}.
Figure~\ref{fig_translationSpecification} uses a runtime environment $T$ that maps variables to memory addresses.
The $n$-th declared variable $x$ will be mapped to its memory address, e.g., such as $T(x) \rightarrow address(n)$.
Whenever a new variable or function is declared, the environment $T$ is updated to a new $T'$ that includes this definition.
This mapping is necessary to distinguish variables that have the same name, when generating code (as seen in Figure~\ref{fig_translationSpecification}), and when generating certificates.

\begin{figure}[tbh]
\centering
\input{assets/translationSpecification.tex}
\caption{Specification of the translation rules.}
\Description{Specification of the translation rules.}
\label{fig_translationSpecification}
\end{figure}

\subsection{The Numbering Schema}
\label{sub_numbering_schema}

Given a \textsc{CharonLang} program $P$ written with $n$ symbols, its certificate $\mathit{Cert}(P)$ is the product of the first $n$ primes, each raised to a particular exponent.
Each symbol uses a different exponent, as we have already illustrated in Examples~\ref{ex_exampleCertificate} and~\ref{ex_godel}.
Figure~\ref{fig_symbolTable} shows the mapping of symbols to exponents.

\begin{example}
\label{ex_if_then_else}
Figure~\ref{fig_symbolTable} shows that an \textsc{If-Then} construct will be always associated with a triple of exponents $(41, 43, 47)$.
These exponents indicate, respectively, the condition the if-then-else command will evaluate, the start of the conditional code block, and the end of the conditional code block.
In regards to if-then-else statements, we have chosen to assign explicit markers to the beginning and to the end of these blocks so the certificate can be inverted back into the original program.
Section~\ref{sub_canonical} discusses inversion in more detail.
\end{example}

\begin{figure}[tbh]
\centering
\begin{footnotesize}
\input{assets/symbolTable.tex}
\end{footnotesize}
\caption{Mapping of programming constructs to the exponents used in the numbering schema.}
\Description{Symbol derivation rules.}
\label{fig_symbolTable}
\end{figure}

In contrast to the encoding of control-flow statements, the encoding of constants, variables, and functions vary depending on which symbols are being referenced.
These cases are discussed below:
\begin{description}
\item [Constant:] A constant $c$ is associated with the exponent $11^{c'}$, where $c'$ is $c + 1$ if $c \geq 0$, or $c$ otherwise.
Therefore, each constant is represented by a unique value.
Adding one is necessary to avoid the identity case of the exponentiation.

\item [Variable definition:] Variables definitions encoded with the base number $13$ to the power of $\mathit{var.\ def.\ exp.}$, following the formula $\mathit{var.\ def.\ exp.} = \mathit{type\ symbol}_1^{\mathit{type\ symbol}_2^{...}}$, where $\mathit{type\ symbol}_n$ is the symbol associated with the type of the $\mathit{n}$-th element of the variable.
Scalars will have a single $\mathit{type\ symbol}$.
Function parameters are encoded similarly, but use $23$ as the base number.

\item [Variable usage:] Variable usage cases employ the base number $17$, to the power of $\mathit{var.\ usage\ exp.}$, which considers the variable prime $\mathit{vp}$ and the memory offset of this element from the variable's base address.
This exponent comes from the formula $\mathit{var.\ usage\ exp.} = \mathit{vp}^{\mathit{mem.\ offset}}$.
If the memory offset is static (i.e., indexing an array with a constant, or accessing an element from a \textit{struct}), then $\mathit{mem.\ offset} = 2^{\mathit{os} + 1}$, where $\mathit{os}$ is the offset in bytes.
Scalar variables are covered by this case and have $\mathit{os} = 0$.
If the memory offset is dynamically determined (as in indexing an array with a variable), then $\mathit{mem.\ offset} = 3^{\mathit{vp}_{id}}$, where $\mathit{vp}_{id}$ is the variable prime of the index.

\item [Function definition:] Function definitions are associated with a pair of symbols: \textsc{Func. (start)} and \textsc{Func. (end)}.
These symbols explicitly delimit the beginning and the end of the function scope.
While the latter uses $37$ as the base number, the former is represented by $31^{\mathit{func.\ def.\ exp.}}$, following the formula $\mathit{func.\ def.\ exp.} = \mathit{type\ symbol}(\theta)^{(\mathit{num.\ params.}(f) + 1)}$ where $\mathit{type\ symbol}(\theta)$ is the symbol associated with the function type, and $\mathit{num.\ params.}$ is the number of parameters the function takes.
$\mathit{num.\ params.}$ is incremented by 1 to avoid the identity case of exponentiation for the case when a function takes no parameters.

\item [Function call:] Function calls are encoded with a base number, $29$, to the power of the function prime $\mathit{fp}_m$ that corresponds to the called function.
\end{description}

\subsubsection{Identification of variables and functions}
\label{sss_variable_and_function_primes}

The Numbering Schema uses unique prime numbers as a means to identify variables and functions.
These numbers are determined by the variable's or function's declaration site, in a way that resembles De Bruijn's indexation~\cite{Brijn94}.

We collect all the variables definitions --- including function parameters ---, in the order they are defined, and map the $\mathit{n}$-th active variable in this set to the $\mathit{n}$-th prime number $\mathit{vp}_n$.
A variable is active if it has at least one use case, and we treat function parameters as always active.
The prime of a variable is emitted as soon as an use case of it is found.

Function primes $\mathit{fp}_m$ are produced by mapping the $m$-th function label to the $m$-th prime number $\mathit{fp}_m$.
These numbers are emitted as soon as the function definition is parsed by the compilation algorithm $\mathit{Comp}$.
Unlike variable primes, that only consider active variables, function primes are emitted regardless of whether the function is ever called.

\begin{example}
Consider the set of variables definitions $\{ \mathit{var_1}, \mathit{var_2}, \mathit{var_3} \}$.
Assume that $\mathit{var_1}$ and $\mathit{var_3}$ both have use cases, but $\mathit{var_2}$ does not.
In this case, $\mathit{var_1}$ is mapped to the first variable prime $\mathit{vp}_1 = 2$, while $\mathit{var_3}$ is mapped to the second variable prime $\mathit{vp}_2 = 3$.
\end{example}

Variable primes $\mathit{vp}$ and function primes $\mathit{fp}$ are managed by a certification environment $\mathit{C}$.
The high-level certification algorithm uses the $\mathit{C}_H$ environment, which will be discussed in Section~\ref{sub_source_code}.
$\mathit{C}_H$ maps variables' aliases to variable primes, and function names to function primes.
The low-level certification algorithm, on the other hand, is described in Section~\ref{sub_machine_code} and uses $\mathit{C}_L$.
$\mathit{C_L}$ maps variables' base addresses to variable primes, and the function labels to function primes.
While these environments are fully independent, Figure~\ref{fig_varAndFuncPrime} presents how they produce variable and function primes from the compilation process.

\begin{figure}[tbh]
\centering
\input{assets/varAndFuncPrime}
\caption{Emission of variable and function primes $\mathit{vp}_n$ and $\mathit{fp}_m$.}
\Description{Emission of variable and function primes $\mathit{vp}_n$ and $\mathit{fp}_m$.}
\label{fig_varAndFuncPrime}
\end{figure}

\begin{example}
\label{ex_numberingSchema}
Figure ~\ref{fig_numberingSchema} shows how the Numbering Schema handles functions and variables.
As \texttt{func\_1} is the first function declared in this program, it will be identified by the first function prime, $\mathit{fp}_1 = 2$.
The next declared functions (\texttt{func\_2} and \texttt{main}) are identified by $\mathit{fp}_2 = 3$ and $\mathit{fp}_3 = 5$, respectively.
A similar logic applies to variables definitions, but, in this case, their types are also taken into consideration.
The first variable to be declared is \texttt{int param\_1}, from \texttt{func\_1}.
It has $\mathit{type\ symbol} = 3$ and will be identified by the first variable prime $\mathit{vp}_1 = 2$.
The next variables are \texttt{float\ param\_2} ($\mathit{vp}_2 = 3$, $\mathit{type\ symbol} = 5$), and \texttt{int\ some\_var[3]} ($\mathit{vp}_3 = 5$, $\mathit{type\ symbol} = 7^{7^{3}}$), and \texttt{int\ var\_2} ($vp_{5} = 11$, $\mathit{type\ symbol} = 3$).
\texttt{some\_var} is an array.
Only its third position is used in the program, and it is not indexed via variable indices; hence, only the third position consumes a prime in the certification process.
The last line of the example shows an expression that takes a variable and the result of a function call as operands.
To encode it, the Numbering Algorithm will use the addition base symbol ($79$), and will produce symbols to represent the variable usage and the function call.
The variable usage symbol will be the base $17$ to the power of $\mathit{vp}^{\mathit{mem. offset}}$, where $\mathit{mem. offset} = 2^{0 + 1} = 2$ and $vp_5 = 11$.
Rearranging, the symbol will be $17^{11^{2}} = 17^{121}$.
The function call symbol, in turn, will be $29^{\mathit{fp}_m} = 29^3$.
The positional primes of each component have been omitted for brevity.
\end{example}

\begin{figure}[htb]
    \centering
    \includegraphics[width=1\linewidth]{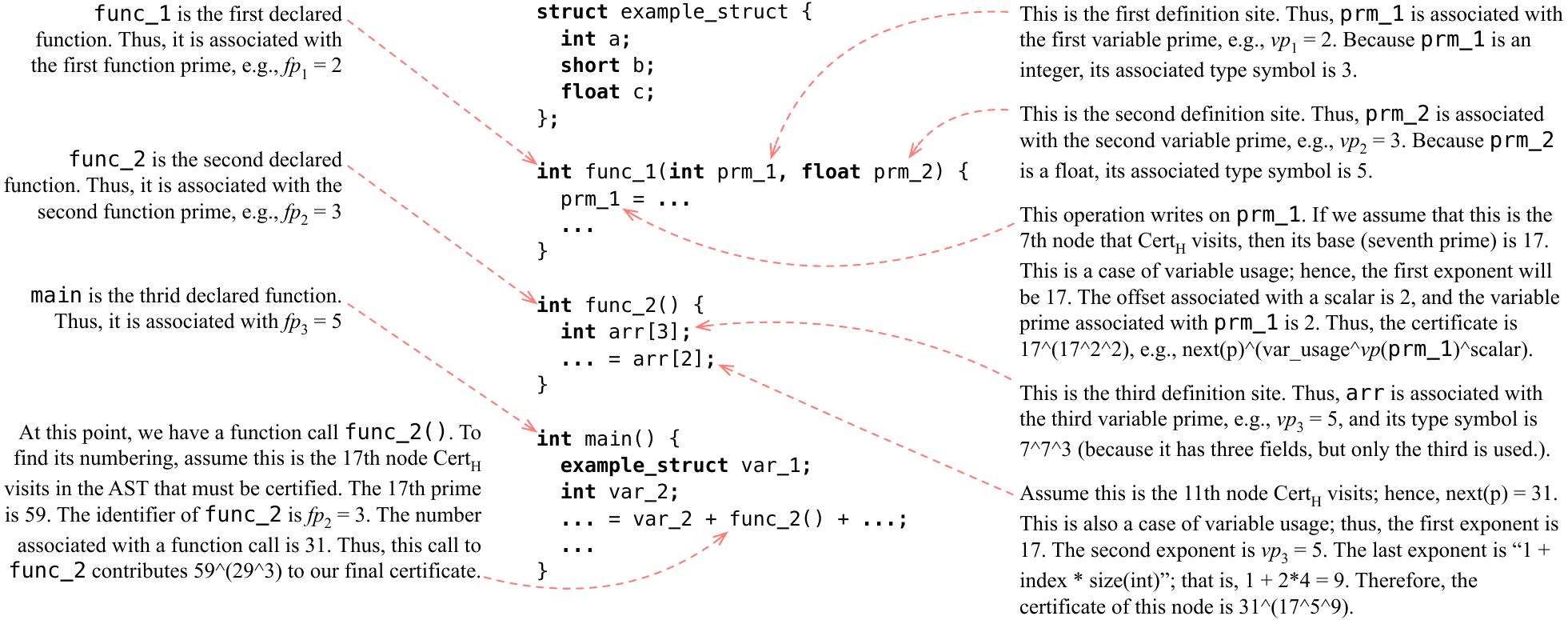}
    \caption{Example illustrating the certification of variables and functions.}
    \Description{Example illustrating the certification of variables and functions.}
    \label{fig_numberingSchema}
\end{figure}

By associating a different prime to each variable, the certification schema is able to distinguish programs that access the same variables, but in different order.
We notice that this schema does not take commutativity into consideration.
Thus, $x + y$ and $y + x$ will produce different certificates.
This property applies to both, the high-level certifier and the low-level certifier, for, according to Figure~\ref{fig_numberingSchema}, both associate user-defined symbols with primes via the environments $C_H$ and $C_L$.

\subsection{Source Code Certification}
\label{sub_source_code}

The certification of programs written in the high-level language (\textsc{CharonLang}) happens in two phases.
First, an environment $C_H$ is created following the rules in Figure~\ref{fig_varAndFuncPrime}.
Then, the certificates are produced by a function $\mathit{Cert}_H$, following the rules in Figure~\ref{fig_highLevelCert}.
Notice that these rules are parameterized by the environment $C_H$, and that we do not pass $C_H$ as a parameter to $\mathit{Cert}_H$, because it never changes; rather, we treat $C_H$ as an immutable global table.

In Figure~\ref{fig_highLevelCert}, $p$ represents the current positional prime, and the other values follow the Numbering Schema seen in Section~\ref{sub_numbering_schema}.
The certifier $\mathit{Cert}_H$ traverses the program's abstract syntax tree, assigning numbers to nodes in post-order.
This is motivated by the fact that the compiler will generate code for the operands before emitting the instructions of the operations themselves.
By certifying expressions in post-order, we conciliate the sequence of numbers produced by $\mathit{Cert}_H$ with the sequence produced by $\mathit{Cert}_L$.

After traversing the AST and computing the encoding numbers for all of its nodes, $\mathit{Cert}_H$ will move all the variable and parameter definition symbols to the beginning of the certificate in the order they appear in the original program, and recompute all of the positional primes.
The relative position of the remaining numbers is also preserved.
This is a key design choice that enables the inversion of certificates\footnote{Adding symbols such as \textsc{Cond}, \textsc{If (start)}, and \textsc{If (end)} to explicitly mark the boundaries of each part of conditionals when handling them follows the same motivation.}.
Inversion is discussed in Section~\ref{sub_canonical}.

\begin{figure}[htb]
    \centering
    \begin{footnotesize}
    \input{assets/highLevelCert.tex}
    \end{footnotesize}
    \caption{Rules for the certification of programs in the high-level language \textsc{CharonLang}.}
    \Description{Source Code Certification rules.}
    \label{fig_highLevelCert}
\end{figure}

\label{ex_source_code_cert}
\begin{figure}
    \centering
    \includegraphics[width=1\linewidth]{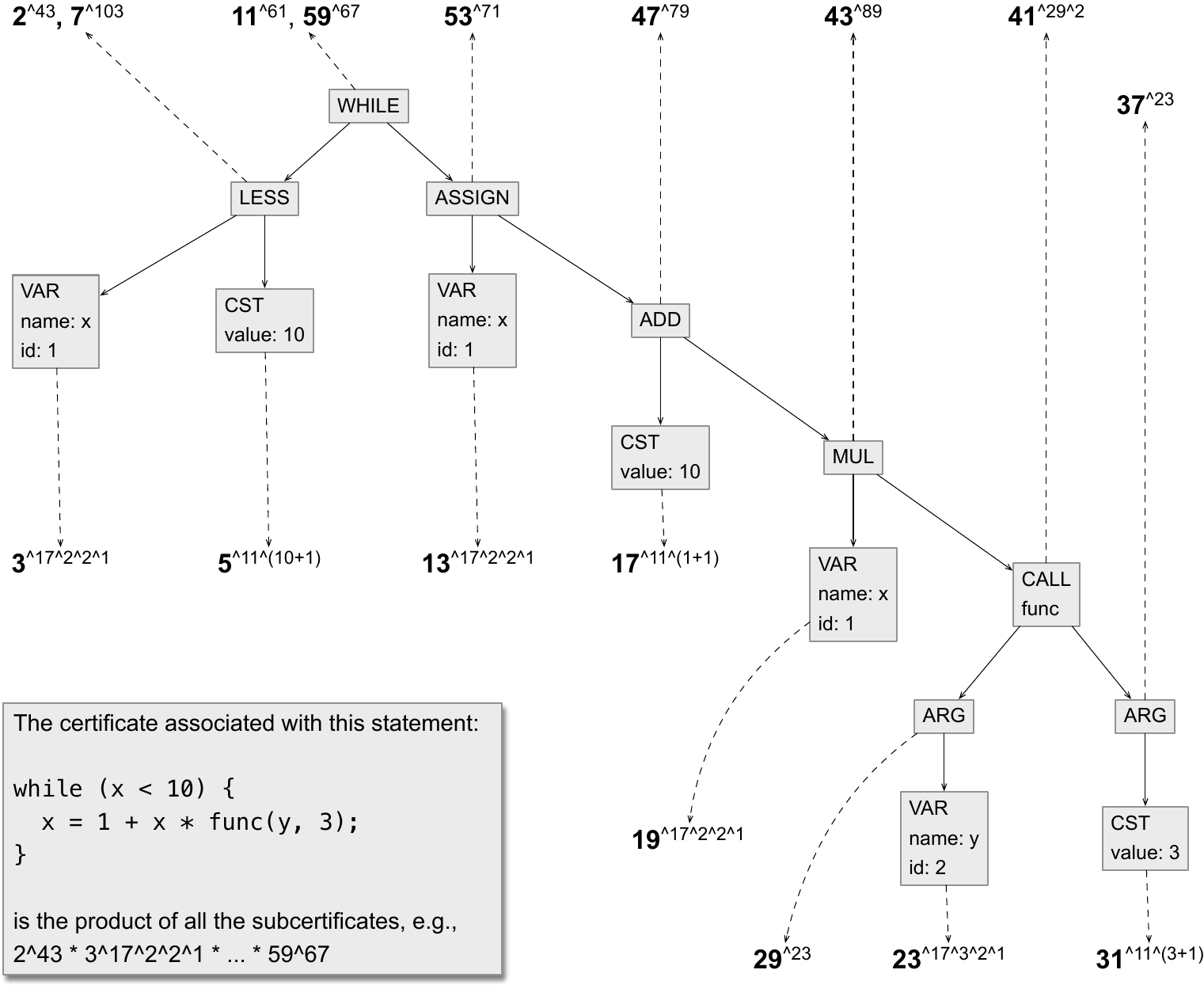}
    \caption{Certification of a \textsc{While} statement, i.e., $Cert_H(\textbf{while} (e) \ \{ S \}) ::= p^{43} \times Cert_H(e) \times next(p) ^ {61} \times Cert_H(S) \times next(p) ^ {67}$.}
    \Description{Certification of a \textsc{While} statement.}
    \label{fig_astCert}
\end{figure}

\begin{example}
\label{ex_astCert}
Figure~\ref{fig_astCert} shows the certificate produced for a \textsc{While} statement.
This example assumes that $x$ has $\mathit{vp}_{1} = 2$, $y$ has $\mathit{vp}_{2} = 3$, and
\texttt{func} has $\mathit{fp}_{1} = 2$.
Following the rules in Figure~\ref{fig_highLevelCert}, $\mathit{Cert}_H$ first certificates the conditional in the loop, which is the binary operation $x < 10$.
In this case, $\mathit{Cert}_H$ will first emit the \textsc{Cond} symbol, and then compute the certificate of the left-hand side, the reading of variable $x$; the right-hand side, the constant $10$; and finally the binary \texttt{LESS} operation itself.
After that, $\mathit{Cert}_H$ emits the number associated with the start of a \textsc{While} construct, which ends up as $11^{61}$, i.e., the fifth prime raised to 61.
In this case, 61 is the exponent associated with the beginning of a \textsc{While} block.
Next, $\mathit{Cert}_H$ proceeds to build a certificate to the body of the while operation, similarly to what we saw in Example~\ref{ex_numberingSchema}.
Finally, it computes the number that represents the end of the conditional code, $59^{67}$.
The exponent 67 is associated with the end of a \textsc{While} block.
\end{example}

\subsection{Machine Code Certification}
\label{sub_machine_code}

The Machine Code Certification algorithm $\mathit{Cert}_L$ produces a certificate for the compiled program, which exists in the low-level \textsc{CharonIR}.
The compiled program consists of a list of instructions for the Virtual Machine to execute, and a set of metadata --- the \texttt{data} section, as previously introduced.
As in the high-level case, the certification of low-level code happens in two phases.
The first creates the environment $C_L$ (Figure~\ref{fig_varAndFuncPrime}); whereas the second builds the certificate itself (Figure~\ref{fig_lowLevelCert}).

As already mentioned, Figure~\ref{fig_varAndFuncPrime} shows the first phase of the certification process.
The rules in Figure~\ref{fig_varAndFuncPrime} produce a map between the base memory address of each variable and a unique variable prime $\mathit{vp}$, and add it to the runtime environment $\mathit{C_L}$.
This environment also contains mappings between each function label to an unique function prime $\mathit{fp}$.

The second phase of the certification process, described in Figure~\ref{fig_lowLevelCert}, performs five actions:

\begin{enumerate}
\item \textbf{Map Dependencies Between Instructions And Temporary Registers:}
Produce a mapping between instructions and the temporary registers they depend upon.
For instance, it will map the instruction of a binary operation to the temporary registers that contains its operands.
This process is recursive.
As such, if one of the operands of a binary operation $\mathit{Op}_1$ is the result of another binary operation $\mathit{Op}_2$, $\mathit{Op}_1$ dependency mapping will also include the registers $\mathit{Op}_2$ depends upon.

\item \textbf{Infer Variables Types:}
Produce a mapping between variables addresses and their inferred types.
Types are inferred based on the sequence of instructions used to read or write a variable.
Reading from and writing to \texttt{float} variables will use the \texttt{LOADF}/\texttt{STOREF} instructions, respectively.
The integer types, \texttt{int} and \texttt{short}, will use \texttt{LOAD} and \texttt{STORE} --- the latter will also be truncated, with the \texttt{TRUNC} instruction, immediately after being read or before being written.

\item \textbf{Analyze Functions:}
Produce a mapping between functions labels and their inferred return types and number of parameters.
Function type inference is computed from inferring the type of the value they return, following the same logic applied to infer variables types.
The number of parameters it takes, on the other hand, is obtained from the counting the occurrences of the parameter passing pattern\footnote{\texttt{CONSTANT R$_c$ var$_{\mathit{address}}$}; \texttt{STORE R$_c$ R$_\mathit{arg}$}} within the function scope.
This pass will also set up the insertion points of the \textsc{Func. (start)} and \textsc{Func. (end)} symbols in the program certificate, based on the scope each function label delimits.

\item \textbf{Analyze Conditional Jumps:}
Search for conditional jumps (i.e., the $\mathtt{JZ}$ instruction) and determine the first instruction that is used in its condition.
This pass is based on the previously computed dependency mapping, and it sets up the insertion point of the \textsc{Cond} symbol in the program certificate.

\item \textbf{Certificate Instructions:}
Visit each instruction of the machine code to compute the adequate number following the rules if Figure~\ref{fig_lowLevelCert}.
This figure presents the patterns the low-level certifier $\mathit{Cert}_H$ will attempt to match.
The certifier will then proceed to the first instruction after the instructions in the pattern it just matched until it reaches the end of the program.
This follows the design of certificating constructs, rather than individual instructions.
\end{enumerate}

\begin{figure}[htb]
    \centering
    \begin{footnotesize}
    \input{assets/lowLevelCert}
    \end{footnotesize}
    \caption{Machine Code Certification rules.}
    \Description{Machine Code Certification rules.}
    \label{fig_lowLevelCert}
\end{figure}

\label{ex_machine_code_cert}
\begin{figure}[htb]
    \centering
    \includegraphics[width=1\linewidth]{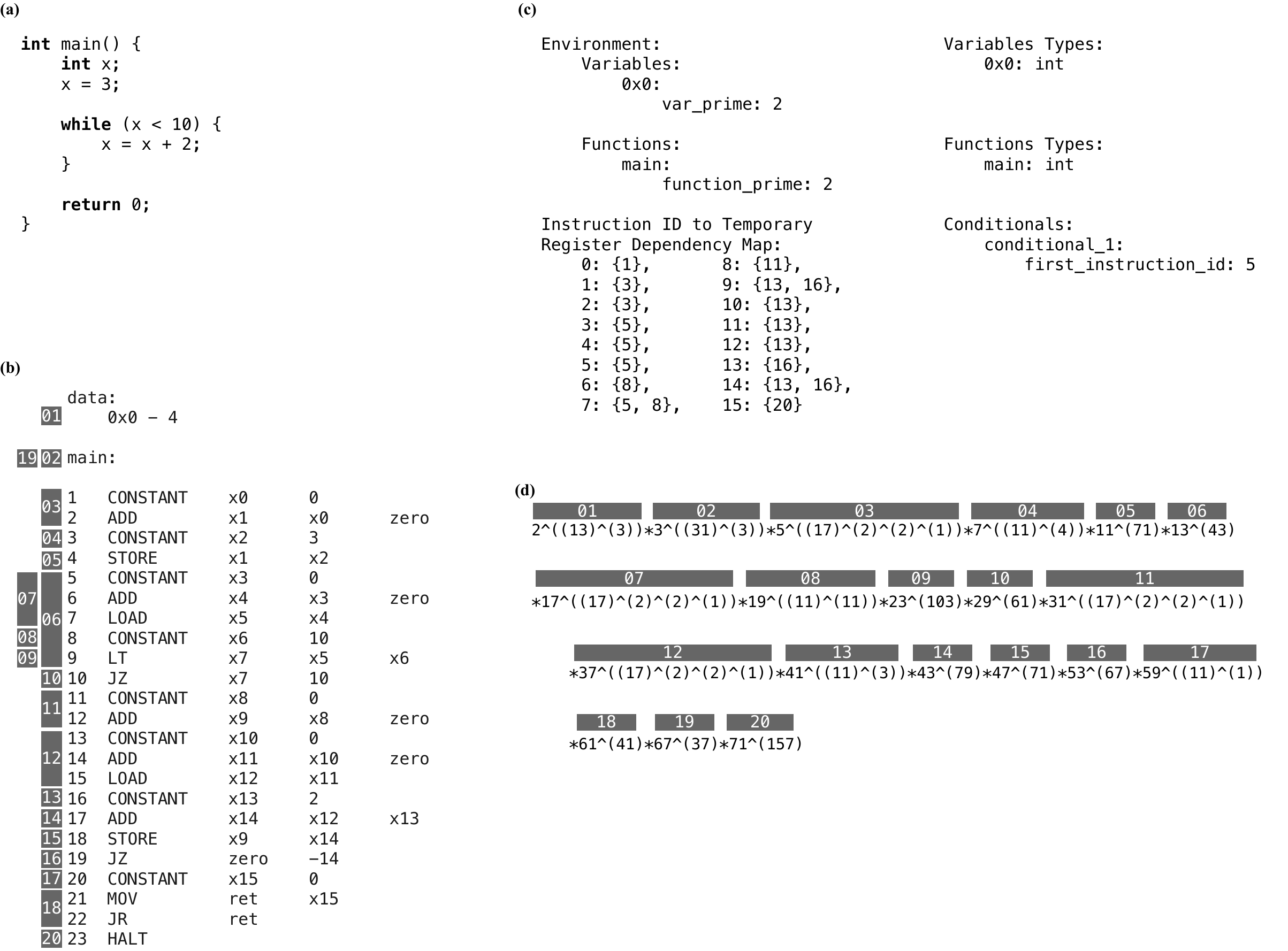}
    \caption{Certification of a simple program with a \textsc{While} loop.}
    \Description{Certification of a simple program with a \textsc{While} loop.}
    \label{fig_machineCodeCert}
\end{figure}

\begin{example}
\label{ex_machineCodeCert}
Figure~\ref{fig_machineCodeCert} (a) displays the a simple program implemented in \textsc{CharonLang} that contains a \textsc{While} loop.
Part (b) of that figure presents its low-level representation in \textsc{CharonIR}.
The figure groups sequences of instructions into blocks to highlight the patterns that the certifier has matched to produce the final certificate.
These blocks are mapped to part (d), which presents the certificate that encodes this program.
Part (c) presents the environment $\mathit{C_L}$ after it is populated by the rules in Figure~\ref{fig_varAndFuncPrime}, and the resulting maps of passes (1) through (3).
\end{example}

\subsection{The Canonical Format}
\label{sub_canonical}

Certificates are invertible.
Thus, if $n = \mathit{Cert}_H(P) = \mathit{Cert}_L(\mathit{Comp}(P))$, then it is possible to reconstruct a program $P_c$ from $n$.
These two programs, $P$ and $P_c$ are semantically equivalent\footnote{In this paper, we do not show semantic equivalence between high- and low-level constructs: we only show proof that the compilation of one will yield the other. In our work, the only difference between programs that lead to the same canonical form is the declaration and storage of variables. Thus, it is simple to conclude that they would implement the same semantics.}; however, they are not syntactically equivalent.
The program $P_c$ is called the {\it canonical} representation of $P$.
As such, the canonical form $P_c$ has the same certificate as the program $P$ it was reconstructed from.
Notice that different programs might have the same canonical representation.
If two programs have the same canonical representation, then they implement the same semantics.
However, these programs might differ in how data is allocated.
The following differences are possible:

\begin{description}
\item [Variables:] All the variables are global, except for function parameters.
The canonical form preserves the order in which variables are declared in the original program.

\item [Structures:] The canonical representation only considers data structures elements that are used in the program.
However, it preserves the number of elements and their order in the original program.
Thus, unused elements in the original program are still defined in the canonical representation, but they are typed as $\mathtt{\_\_unknown\_type\_\_}$.
Arrays that are only indexed by constants in the original program are represented as structs in the canonical format.
Arrays that are indexed with a variable, on the other hand, are still declared as arrays in the canonical format.


\end{description}

\label{ex_canonical_flow}
\begin{figure}[htb]
    \centering
    \includegraphics[width=1\linewidth]{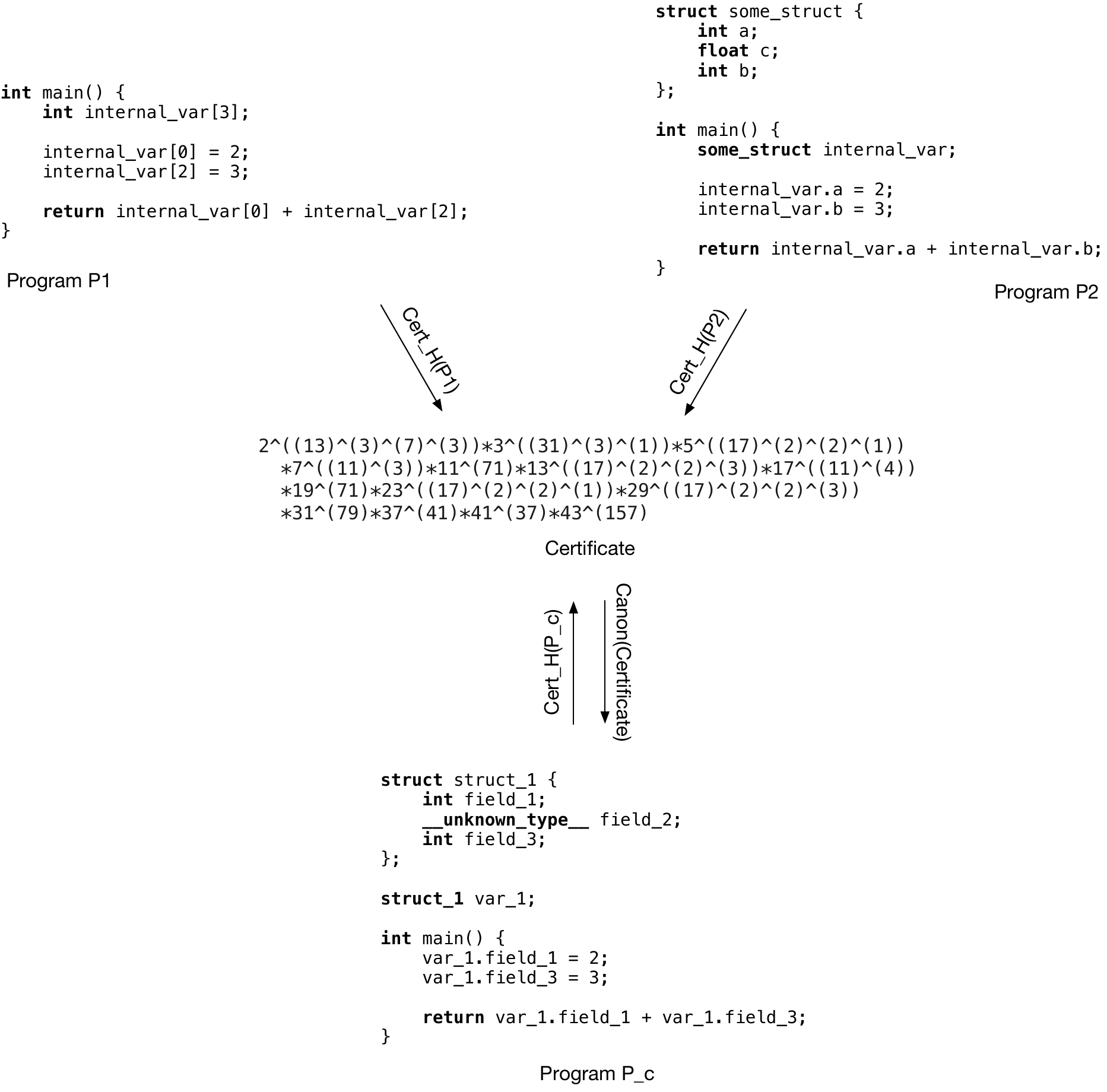}
    \caption{Inversion of the certificate of two semantically equivalent high-level programs.}
    \Description{Inversion of the certificate of two semantically equivalent high-level programs.}
    \label{fig_canonical_flow}
\end{figure}

\begin{example}
\label{ex_canonical}
Figure~\ref{fig_canonical_flow} presents two high-level programs, $P1$ and $P2$, such that $\mathit{Cert}_H(P1) = \mathit{Cert}_H(P2)$.
While these programs are not syntactically identical, they possess the same semantics: they set the values of the first and third elements of a data structure --- both integers --- using constant offsets and return their sum.
These programs differ only on how variables are stored.
The second program stores variables as fields of a struct.
The first does the same, albeit using an array.
The canonical representation uses a struct to represent both situations.
However, one of the fields of this struct is left undefined.
Notice that this field still occupies 4 bytes, as this is the default alignment that we use in \textsc{CharonLang}.
Given that this field is never used, it is also not declared, meaning that it could be represented as any type that fits into 4 bytes.
In this case, this element's type is set to \texttt{\_\_unknown\_type\_\_} in the canonical representation, which occupies one word in the memory layout.
\end{example}

\subsubsection{The Canonical Reconstruction}
\label{sss_inv}

Given a certificate $n$, it is possible to obtain a canonical program $P_c$, such that $\mathit{Cert_H}(P_c) = n$.
This reconstruction process is defined by a function $\mathit{Canon}(n)$, which Algorithm~\ref{algo_inv_structural} shows.
This reconstruction process happens in four stages:

\begin{enumerate}
\item \textbf{Factorization:} Factorize the certificate $n$ into a product of primes in the form $p^{\mathit{exp}}$, where $p$ is the positional prime and $\mathit{exp}$ is the construct-encoding exponent.
This exponent is expected to follow the symbols and rules from~\hyperref[sub_numbering_schema]{Numbering Schema}.

\item \textbf{Analysis:} This step decides whether a data structure --- i.e., a variable whose definition has more than one type symbol in $\mathit{exp}$ --- should be defined as an array or as a struct.
This is achieved by analyzing all the use cases of these variables.
If the variable's attributes are only accessed with static offsets, it will be defined as a struct; $\mathit{Canon}$ defines it as an array otherwise.

\item \textbf{Generation:} Iterate once more over the factorized certificate to produce the canonical program by matching the pattern from the Numbering Schema, and applying the rules discussed above.

\item \textbf{Finalization:} As there is no explicit symbol for the \texttt{main} function in the certificate, and \textsc{CharonLang} assumes all the programs will follow a def-use relation, $\mathit{Canon}$ will set the last function it defines to be the \texttt{main} function.
\end{enumerate}

Algorithm~\ref{algo_inv_structural} shows the pseudo-code of the reconstruction process for expressions, and it is implemented by the procedure $\mathit{Canon}$.
As a simplification, it omits the handling of functions and control-flow constructs.
The rest of this section illustrates how the reconstruction process in Algorithm~\ref{algo_inv_structural} works through a series of examples, following Figure~\ref{fig_canonicalizer}.
This figure shows the canonical reconstruction of the certificate $n$ from Figure~\ref{fig_canonical_flow}.

\label{ex_canonicalizer}
\begin{figure}[htb]
    \centering
    \includegraphics[width=1\linewidth]{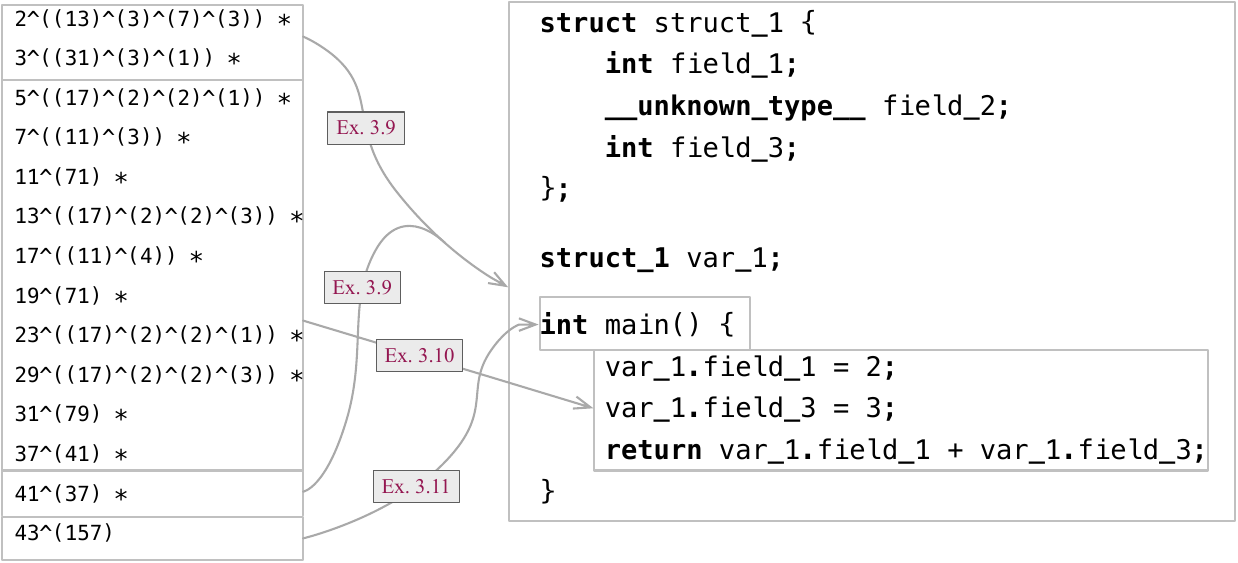}
    \caption{The canonical reconstruction of a simple program.}
    \Description{The canonical reconstruction of a simple program.}
    \label{fig_canonicalizer}
\end{figure}

\paragraph{$\mathit{Canon}$ Parses Certificates Left-to-Right.}
The canonical program generation begins by parsing variable definitions ($\mathit{exp.} = 13^{\mathit{var.\ def.\ exp.}}$) and function parameters ($\mathit{exp.} = 19^{\mathit{var.\ def.\ exp.}}$).
These are the first constructs encoded in the certificate, if any, as discussed in section~\ref{sub_source_code}.
Variable definitions will be added to the program body immediately.
Function parameters, on the other hand, are first added to the list of function parameters, and then passed to the program when $\mathit{Canon}$ parses a \textsc{Func. (start)} symbol ($\mathit{exp.} = 31^{\mathit{type\ symbol}(\mathit{type}(f))^{(\mathit{num. params.}(\mathit{f}) + 1)}}$).
In this case, it will take the first $n$ elements from this list and add to the function definition, where $n$ is the number of parameters this function takes obtained from the symbol.
$\mathit{Canon}$ will name the $n$-th declared variable or parameter \texttt{var\_n}, and the $m$-th defined function \texttt{func\_m}, and will  retrieve its type from $\mathit{exp}$.
It will use these names to refer to the variable or parameter with the $n$-th variable prime, or the function with the $m$-th function prime in the certificate factors that follow.

\begin{example}
\label{ex_canonicalReconstruction_1}
The certificate in Figure~\ref{fig_canonicalizer} encodes a single variable definition, \texttt{internal\_var}.
Its factor is $2 \textasciicircum ((13) \textasciicircum (3) \textasciicircum (7) \textasciicircum (3))$, and it tells this variable is a data structure as there are three type symbols in it ($3$, $7$, and $3$).
$\mathit{Canon}$ will search for variable use cases with $2$ (the first prime) as the variable prime because this is the first variable to be declared.
There are four instances: $5 \textasciicircum((17) \textasciicircum(2) \textasciicircum(2) \textasciicircum(1))$, $13 \textasciicircum((17) \textasciicircum(2) \textasciicircum(2) \textasciicircum(3))$, $23 \textasciicircum((17) \textasciicircum(2) \textasciicircum(2) \textasciicircum(1))$, and $29 \textasciicircum((17) \textasciicircum(2) \textasciicircum(2) \textasciicircum(3))$.
The fields of the variable are accessed with a constant for index in all of the instances; as such, this variable will be set to be declared as a struct with an integer, an element of unknown type, and another integer --- in this order.
The next factor, $3 \textasciicircum ((31) \textasciicircum (3))$, marks the beginning of a function, which is first defined as \texttt{func\_1}.
This function scope will be closed when $\mathit{Canon}$ reaches factor $41 \textasciicircum 37$.
\end{example}

\paragraph{Expressions in Reverse Polish Notation.}
Expressions are parsed following a simple Reverse Polish Notation (RPN) inversion algorithm.
Operands --- variables ($\mathit{exp.} = 17^{\mathit{var.\ usage\ exp.}}$), constants ($\mathit{exp.} = 11^{c'}$), and function calls ($\mathit{exp.} = 29^{\mathit{fp}}$) --- are always added to the expression stack.
Function calls are also parsed following this notation: the algorithm will first parse the arguments ($\mathit{exp.} = 23$), and then the function call itself.
When $\mathit{Canon}$ finds a certificate factor that represents an operator, it will take the top elements (one if it is an unary operator, or two if it is a binary operator) from the stack, write the parenthesized expression, and add it back to the top of the stack.
The expression is only added to the program body when the algorithm finds a terminator.
Any factor that does not encode an operand or an operator is a terminator.
The assignment operator is an exception for this rule, for it terminates an expression in RPN.

\begin{example}
\label{ex_canonicalReconstruction_2}
Following, $5 \textasciicircum ((17) \textasciicircum (2) \textasciicircum (2) \textasciicircum (1))$ encodes a variable usage and causes $\mathit{Canon}$ to start analyzing an expression.
This certificate factor has $2$ as the variable prime; as such, it refers to \texttt{var\_1}.
Given that this variable is a struct, and that the certificate factor encodes the access of the first element, $\mathit{Canon}$ will add \texttt{.field\_1} as element access suffix to the variable name it will push to the expression stack.
$7 \textasciicircum ((11) \textasciicircum (3))$ comes right after it, and it encodes a constant.
As the associated constant is positive, $\mathit{Canon}$ will subtract 1 from it before pushing it to the expression stack.
The next factor, $11 \textasciicircum (71)$, is the assignment operator, and the algorithm will write \texttt{var\_1.field\_1 = 2;}.
As it is a terminator, this expression will be immediately added to \texttt{program}.
This iterative process continues until $\mathit{Canon}$ reaches the last factor, $43 \textasciicircum (157)$.
\end{example}

\paragraph{Reconstruction of Control Flow.}
Control-flow constructs are computed in $\mathit{Canon}$ by first parsing the expression that follows a \textsc{Cond} factor ($\mathit{exp.} = 43$), until it finds the pattern that marks the beginning of a conditional code block (eg., \textsc{If (start)}, which has
$\mathit{exp} = 47$).
Then, it will add the control-flow construct followed by the conditional expression to the program body.
Function and control-flow scopes are closed as soon as the algorithm finds the corresponding ending marker (eg., \textsc{If (end)}, which has $\mathit{exp.} = 53$).

\begin{example}
\label{ex_canonicalReconstruction_3}
As $\mathit{Canon}$ finds the program-ending factor ($\mathit{exp.} = 157$), it will add the \texttt{main} function to the program it just created.
It will rename the last function it defined --- in this case, \texttt{func\_1} --- to \texttt{main}.
Finally, it returns the canonical program $\mathit{P_c}$.
\end{example}

\begin{algorithm}[t!]
\caption{\textsc{Canon} — The Reconstruction of Canonical Forms\label{algo_inv_structural}}
\begin{small}
    \begin{algorithmic}[1]
\STATE \textbf{Input:} Certificate $n$
\STATE \textbf{Output:} Canonical program $P_c$
\STATE

\STATE \textbf{function} \textsc{Canon}($n$)
\begin{ALC@g}
    \STATE \textbf{return} \textsc{CanonRec}(F $\leftarrow$ \textsc{Factorize}($n$),
    P $\leftarrow$ empty program,
    S $\leftarrow$ empty stack,
    v $\leftarrow$ 0)
\end{ALC@g}

\STATE

\STATE \textbf{function} \textsc{CanonRec}($F$, $P$, $S$, $v$)
\begin{ALC@g}
    \IF{$F$ is empty}
        \STATE $P.\textsc{AddMain}()$
        \STATE \textbf{return} $P$
    \ENDIF

    \STATE $\text{f, rest} \leftarrow F[0], F[1:]$
    \STATE $\text{e} \leftarrow \textsc{GetExpFromFactor}(\text{f})$

    \IF{\textsc{IsVarDef}$(\text{e})$}
        \STATE $v \leftarrow v + 1$
        \STATE $\text{name} \leftarrow \texttt{``var\_''} + v$
        \STATE $\tau \leftarrow \textsc{GetVarType}(\text{e})$

        \IF{\textsc{IsSingleType}$(\tau)$}
            \STATE $P.\textsc{AddGlobalVariable}(\text{name}, \tau)$
        \ELSIF{\textsc{IsStruct}(e)}
            \STATE $\text{fields} \leftarrow \textsc{GetFieldsFromVarType}(\tau)$
            \STATE $P.\textsc{AddStruct}(\text{name}, \text{fields})$
        \ELSE
            \STATE $\text{size} \leftarrow \textsc{GetSizeFromVarType}(\tau)$
            \STATE $P.\textsc{AddArray}(\text{name}, \tau, \text{size})$
        \ENDIF

    \ELSIF{\textsc{IsVarUsage}$(\text{e})$}
        \STATE $\text{name} \leftarrow \textsc{GetVarNameFromVarUsage}(e)$
        \STATE $S.\textsc{Push}(\text{name})$

    \ELSIF{\textsc{IsConstant}$(\text{e})$}
        \STATE $S.\textsc{Push}((c - 1)$ if $c \geq 0$ else $c)$

    \ELSIF{\textsc{IsOperation}$(\text{e})$}
        \STATE $\text{op} \leftarrow \textsc{GetOperationSymbolFromExp}(\text{e})$
        \STATE $\text{rhs} \leftarrow S.\textsc{Pop}()$
        \IF{\textsc{IsUnary}(e)}
            \STATE $\text{expr} \leftarrow (\text{op} + \text{rhs})$
        \ELSE
            \STATE $\text{lhs} \leftarrow S.\textsc{Pop}()$
            \STATE $\text{expr} \leftarrow (\text{lhs} + \text{op} + \text{rhs})$
        \ENDIF
        \STATE $S.\textsc{Push}(\text{expr})$
    \ENDIF

    \IF{\textsc{IsTerminator}$(\text{e})$}
        \STATE $\text{stmt} \leftarrow S.\textsc{Pop}()$
        \STATE $P.\textsc{AddExpr}(\text{stmt})$
    \ENDIF

    \STATE \textbf{return} \textsc{CanonRec}$(\text{rest}, P, S, v)$
\end{ALC@g}
\end{algorithmic}
\end{small}
\end{algorithm}

\subsection{Correctness}
\label{sub_correctness}

This section provides sketches of proofs for the two main results discussed in this paper.
In Section~\ref{sss_correctness_equivalence} we show that compilation preserves certificates.
Then, in Section~\ref{sss_correctnessInvertibility} we show that certificates are invertible.

\subsubsection{Correctness of Equivalence of Certificates}
\label{sss_correctness_equivalence}

To ensure that the certification process is correct, we establish that our compiler preserves the certification result. 
Specifically, for any high-level program $P_H$ compiled to a low-level program $P_L$, the outputs of the certification algorithms $\mathit{Cert}_H$ and $\mathit{Cert}_L$ coincide.
This property requires the $n$-th variable prime $\mathit{vp_n}$ produced by both certification algorithms $\mathit{Cert}_H$ and $\mathit{Cert_L}$ to be equal.
This is proven in Lemma~\ref{lm_variablePrimes}.
It shows that the certification environments $\mathit{C_H}$ and $\mathit{C_L}$ produced for the high- and low-level certification algorithms, respectively, map the high- and low-level representations of a left-hand side term $a$ to the same variable prime.

\begin{lemma}[Equivalence of Variable Primes]
\label{lm_variablePrimes}
If $a$ is a left-hand side term, and $\mathit{Comp}(a) = \ldots T(a)$, then $\mathit{C_H(a)} = \mathit{C_L(T(a))}$.
\end{lemma}

\begin{proof}
The proof goes by case analysis on the rules shown in Figure~\ref{fig_varAndFuncPrime}.
Consider the case where $a = \mathit{v}$, i.e., a variable.
The corresponding rule Figure~\ref{fig_varAndFuncPrime} shows that if $\mathit{v} \notin C_H$, then a fresh variable prime $\mathit{vp}' = \mathsf{next}(\mathit{vp})$ is generated and $C_H$ is updated with $\mathit{v} \mapsto \mathit{vp}'$, while $C_L$ is updated with $T(\mathit{v}) \mapsto \mathit{vp}'$.
Thus, both environments remain equivalent, assigning the same prime to $\mathit{v}$ and to its corresponding address in the low-level domain.

The array access case, $a = \mathit{v}[e]$, happens in two steps.
First, $\mathit{v}$ is compiled under the same rule as above, assigning it a prime if it is not already mapped to one in $C_H$ and $C_L$.
Then, the index $e$ is compiled separately.
If $e$ happens to be a variable, then it will also be subject to the \textit{variable} rule, and the certification environments will be updated accordingly.
If not, then the compilation of $e$ will not change $C_H$ and $C_L$.
As such, any changes to the certification environment will be consistent across $C_H$ and $C_L$.

For $a = \mathit{v}.x$, i.e., a field access in a struct, the rule shows that the base variable $\mathit{v}$ is compiled following the first rule, and then the offset $T(\mathit{v}.x)$ is added to it.
Since $C_H$ and $C_L$ are only updated when $\mathit{v}$ is compiled (and not for the field access itself), consistency is preserved for the assigned variable prime.

Finally, in the case of function definitions, $\theta_f\ f$, the rule specifies that if $f \notin C_H$, a fresh function prime $\mathit{fp} = \mathsf{next}(\mathit{fp})$ is generated.
The environments $C_H$ and $C_L$ are then updated with $f \mapsto \mathit{fp}$ and $T(f) \mapsto \mathit{fp}$, respectively.
Since the same prime is used in both updates, the environments remain equivalent.

In all cases, when a fresh prime is introduced, it is added to both environments simultaneously.
Thus, we have that for all $a$, the relation $\mathit{C_H}(a) = \mathit{C_L}(T(a))$ is true.
\end{proof}

Theorem~\ref{thm_certificationPreservation} is the core result of this work.
It establishes that certification is preserved through compilation: the certificate computed from a high-level program matches exactly the one computed from its compiled low-level form.
This result ensures that certificates abstract program behavior in a way that is invariant under compilation.

\begin{theorem}[Preservation of Certificates]
\label{thm_certificationPreservation}
If $P_H$ is a program in the high-level \textsc{CharonLang} language,
then $\mathit{Cert}_H(P_H) = \mathit{Cert}_L(\mathit{Comp}(P_H))$.
\end{theorem}

\begin{proof}
The proof proceeds by structural induction on the syntax of the source program $P_H$.

\textbf{Base cases:} the certification of constants, variables, array and struct accesses does not involve the recursive certification of subexpressions; thus, the theorem follows without induction:

\begin{description}
\item [Constant:] If $c$ is a constant, then 
$\mathit{Cert}_H(c) = \mathit{Cert}_L(\mathtt{CONSTANT\ R_c\ c}) = 2^{11^c}$

\item [Variable:] If $x$ is a variable name, then
$\mathit{Cert}_H(x) = \mathit{Cert_L}(\mathtt{CONSTANT \ R_c \ x};\ \mathtt{ADD\ R_x \ R_c \ zero}) = p ^ {{17} ^ \mathit{var.\ usage\ exp.}}$.
The exponent $\mathit{var.\ usage\ exp.}$ tells whether this is a simple variable or an element in a data structure --- and whether it is accessed with a static \footnote{As an array with constant for index or element in a struct} or a dynamic offset\footnote{As an array with a variable for index}.
If it is a data structure, the high- and low-level programs $P_H$ and $P_L$ programs will also contain expressions or instructions to calculate the element offset from the base variable.

In either case, $\mathit{Cert}_H$ and $\mathit{Cert_L}$ will have the same environments $\mathit{C_H}$ and $\mathit{C_L}$, as a direct result from Lemma~\ref{lm_variablePrimes}.
As $\mathit{var.\ usage\ exp.}$ depends solely on the expressions and instructions that implement the variable and the certification environments $\mathit{C_H}$ and $\mathit{C_L}$, we can conclude that the high- and low-level certification of variables coincide.
\end{description}

\textbf{Inductive step:} We apply induction to prove preservation of certificates on constructs whose syntax is defined recursively in Figure~\ref{fig_charonGrammar}.
Thus, if we have a construct defined as $S_0 ::= \ldots S_1 \ldots$, then we assume preservation over $S_1$ to demonstrate that it holds for $S_0$.
As an example, we shall consider the compound program construct $P_H = S_0 = \texttt{if}(e)\ \{ S \}$.
According to the compilation rules (Fig.~\ref{fig_translationSpecification}), the compiler emits the following instruction sequence for such conditionals:
\begin{align*}
    \mathit{Comp}(\texttt{if}(e)\ \{ S \})  &= I_e;\ \texttt{JZ}\ R_e\ \ell;\ I_t,\ T' \\
    \mathit{Comp}(e) &\rightarrow I_e,\ R_e \\
    \mathit{Comp}(S) &\rightarrow I_t,\ T'
\end{align*}

From the certificate construction rules (Fig.~\ref{fig_astCert}), we know:
\[
\mathit{Cert}_H(\texttt{if}(e)\ \{ S \}) 
= p^{43} \times \mathit{Cert}_H(e) \times \mathit{next}(p)^{47} \times \mathit{Cert}_H(S) \times \mathit{next}(p)^{53}
\]

Similarly, from the low-level certificate construction rules (Fig.~\ref{fig_lowLevelCert}):
\[
\mathit{Cert_L}(I_e;\ \texttt{JZ}\ R_e\ \ell;\ I_t)
= p^{43} \times \mathit{Cert_L}(I_e) \times \mathit{next}(p)^{47} \times \mathit{Cert_L}(I_t) \times \mathit{next}(p)^{53}
\]

By the inductive hypothesis, we assume:
\begin{align*}
\mathit{Comp}(e) = I_e &\Rightarrow \mathit{Cert}_H(e) = \mathit{Cert_L}(I_e) \\
\mathit{Comp}(S) = I_t &\Rightarrow \mathit{Cert}_H(S) = \mathit{Cert_L}(I_t)
\end{align*}

We now observe that the same structure appears in both $\mathit{Cert}_H$ and $\mathit{Cert_L}$.
From this, we can conclude that the conditional symbol $p^{43}$ and the the control flow symbols $\mathit{next}(p)^{47}$ and $\mathit{next}(p)^{53}$ will be placed identically in both constructions, as they depend exclusively on the structure of the programs.
Therefore,

\[
\mathit{Cert}_H(\texttt{if}(e)\ \{ S \}) = \mathit{Cert_L}(\mathit{Comp}(\texttt{if}(e)\ \{ S \}))
\]

The proof of equivalence between $\mathit{Cert}_H$ and $\mathit{Cert_L}$ for other constructs is analogous.
\end{proof}

Unlike the certification of high-level programs, where $\mathit{Cert}_H(P_{H_1})$ might be equal to $\mathit{Cert}_H(P_{H_2})$ even if $P_{H_1} \neq P_{H_2}$, $\mathit{Cert_L}(P_{L_1})$ will be equal to $\mathit{Cert_L}(P_{L_2})$ if, and only if $P_{L_1} = P_{L_2}$.
This property is demonstrated in Theorem~\ref{thm_uniquenessLowLevel}.
To prove it, we will use an auxiliary result: Lemma~\ref{lm_uniquenessLowLevelFactor}, which demonstrates that each product of certificate factors can only be obtained from unique sequences of instructions.

\begin{lemma}[Uniqueness of Low-level Certificate Factors]
\label{lm_uniquenessLowLevelFactor}
The low-level certification function $\mathit{Cert}_L$ is injective.
\end{lemma}

\begin{proof}
The proof proceeds by analyzing the cases presented in Figure~\ref{fig_lowLevelCert}.
Most cases are trivial because each construct-encoding exponent  (the $p$ on the right side of Figure~\ref{fig_lowLevelCert}) is unique and strictly positive, and only appears in a single rule.
Given that both $p$ and every exponent are prime numbers, the Fundamental Theorem of Arithmetic establishes that these powers produce a unique factorization\footnote{For distinct primes $p \neq q$ and positive integers $m,n$, the equality $p^m = q^n$ is impossible, since by the Fundamental Theorem of Arithmetic the prime factorization of an integer is unique.}.

There are two exceptions for this rule: the rules that map to $p^{17^{\mathit{var.\ usage\ exp.}}}$, and the rules that feature $p^{43} \times \mathit{Cert}_L(S) \times \dots$ in the right-hand side: these exponents, $17^{\mathit{var.\ usage\ exp.}}$ and $43$ appear multiple times on the right side of the rules in Figure~\ref{fig_lowLevelCert}.

The argument for the first exception, $p^{17^{\mathit{var.\ usage\ exp.}}}$, relies on the fact that the exponent $\mathit{var.\ usage\ exp.}$ is also a power of primes that will produce different numbers depending on the applied rule, such that $C_L(x) ^ {2 ^ {1 + \mathit{offset}}} \neq C_L(x)^{3 ^ {C_L(y)}}$.
$C_L(x)$ (and, similarly, $C_L(y)$) is the variable prime that uniquely represents the variable $x$ (or $y$), and $1 + \mathit{offset}$ is also strictly positive.
The result follows from the observation that there are no values for which $2 ^ {1 + \mathit{offset}} = 3 ^ {C_L(y)}$.
As for the second exception, there are three rules where $p^{43}$ is featured:
\begin{enumerate}
\item $S; \ \mathtt{JZ} \ \mathtt{R}_S \ \mathit{len}(S'); \ S'$
\item $S; \ \mathtt{JZ} \ \mathtt{R}_S \ \mathit{len}(S'); \ S'; \ \mathtt{JZ\ zero} \ \mathit{len}(S''); \ S''$
\item $S; \ \mathtt{JZ} \ \mathtt{R}_S \ \mathit{len}(S'); \ S'; \ \mathtt{JZ\ zero} \ - (\mathit{len}(S) + \mathit{len}(S') + 1) $
\end{enumerate}
Even though all of the factors produced by (1) are also produced by (2), the latter is the only rule that can add a factor with exponent $59$.
Similarly, (3) is the only rule that can add factors with exponents $61$ and $67$.
\end{proof}

\begin{theorem}[Uniqueness of Low-level Programs]
\label{thm_uniquenessLowLevel}
Let $P_1$ and $P_2$ be two programs in the low-level \textsc{CharonIR} intermediate representation.
Then $\mathit{Cert}_L(P_1) = \mathit{Cert}_L(P_2)$ if and only if $P_1 = P_2$.
\end{theorem}

\begin{proof}
The \textit{if} case, $\mathit{Cert}_L(P_1) = \mathit{Cert}_L(P_2)$ if $P_1 = P_2$, is trivial because $\mathit{Cert}_L$ is deterministic.
Therefore, passing the same certificate to $\mathit{Cert}$ yields the same result.

To prove \textit{$\mathit{Cert}_L(P_1) = \mathit{Cert}_L(P_2)$ only if $P_1 = P_2$}, we will use structural induction on $n$, where $n = \mathit{Cert}_L(P_1) = \mathit{Cert}_L(P_2)$.
We proceed by analyzing the right side of each rule in Figure~\ref{fig_lowLevelCert}.

\textbf{Base case:} The certification of constants, variables definitions, and variable usage do not cause the recursive invocation of $\mathit{Cert}_L$ on subprograms.
Thus, the equality of these constructs in programs $P_1$ and $P_2$ and can be demonstrated by applying the result of Lemma~\ref{lm_uniquenessLowLevelFactor} without induction.

\begin{description}
\item [Constant:] From Lemma~\ref{lm_uniquenessLowLevelFactor}, $p^{11^{\mathit{c} + 1\ \mathit{if\ c} \geq 0\ \mathit{else}\ c}}$ is the only factor that encode the instruction \texttt{CONSTANT R\_c c}.
As this formula is injective, the equality of certificates implies the same \texttt{CONSTANT} instruction appears in the same position in both programs.

\item [Variable definition:] A variable definition is encoded from an address in \texttt{data} section of a program.
From Lemma~\ref{lm_uniquenessLowLevelFactor}, it is solely certificated by $p ^ {13^{\mathit{type\ symbol}(\mathit{var_{address}})}}$, where $\mathit{type\ symbol}$ encodes the inferred type of this variable.
If both programs have the same certificate, then the same type has been inferred for the same address of \texttt{data}.
This implies programs $P_1$ and $P_2$ declare the exact same variables, defined at the same relative position.

\item [Variable usage:] The variable usage certificate follows the formula $\mathit{p}^{17^{var.\ usage. exp.}}$.
The exponent, $\mathit{var.\ usage\ exp.}$, will vary depending on the variable usage construct it encodes:

\begin{description}
\item [Variable usage with static offset:] This case is implemented in low-level from a sequence of instructions that will load the variable base address into a temporary register, compute the static offset, and add it to the base address.
Lemma~\ref{lm_uniquenessLowLevelFactor} establishes that this is the only case that produces $\mathit{var.\ usage\ exp.} = C_L(\mathit{x})^{2^{1 + \mathit{offset}}}$, where $\mathit{x}$ is the variable's base memory address, $C_L(\mathit{x})$ is its variable prime in the certification environment, and \textit{offset} is its static memory offset.

\item [Variable usage with dynamic offset:] The other case for the variable usage construct also starts by loading the base memory address into a temporary register.
Unlike the previous case, it will load the contents of another variable into another temporary register, and have memory address arithmetic instructions to compute the address of interest in runtime.
Lemma~\ref{lm_uniquenessLowLevelFactor} tells the exponent can only be $\mathit{var.\ usage\ exp.} = C_L(\mathit{x})^{3^{C_L(\mathit{y})}}$, where $C_L(\mathit{x})$ and $C_L(\mathit{y})$ are the variable primes of the variable being used and of the variable to dynamically compute the memory offset with.
$\mathit{x}$ is the base memory address.
\end{description}

If the certificates are identical, then the same sequence of instructions must be present in the same position in both programs, the variables must use the same addresses, and any offsets from the base memory addresses must be identical, in either variable construct the programs implement.
Otherwise, there would be at least one differing variable usage certificate.
\end{description}

\textbf{Inductive step:}
We shall analyze one case, as the others will be similar.
We consider the compound programs $P_1 = S_1; S'_1$.
We have that $\mathit{Cert}_L(P_1) =  \mathit{Cert}_L(S_1) \times \mathit{Cert}_L(S'_1)$.
We have that $\mathit{Cert}_L(S_1) < \mathit{Cert}_L(P_1)$ and
$\mathit{Cert}_L(S'_1) < \mathit{Cert}_L(P_1)$; thus, we can apply induction.
By induction, we assume that the programs $S_1$ and $S_2$ are unique.
Therefore, it follows that if there exists $P_2$ such that $\mathit{Cert}_L(P_1) = \mathit{Cert}_L(P_2)$, then $P_2 = S_1; S'_1$.
A similar result can be derived from the other recursive rules.
\end{proof}

\subsubsection{Correctness of Invertibility}
\label{sss_correctnessInvertibility}

The canonical form is constructed from the certificate produced by either $\mathit{Cert}_H$ or $\mathit{Cert}_L$, as discussed in Section~\ref{sub_canonical}.
The program produced by $\mathit{Canon}$ captures the logical structure implied by the certificate, independent of the specific syntactic choices made in the source or compiled program.
If two programs yield the same certificate, then they must share the same canonical form.
This property is demonstrated in Theorem~\ref{thm_uniquenessCanon}.

\begin{theorem}[Uniqueness of Canonical Form]
\label{thm_uniquenessCanon}
Let $P_{1}$ and $P_{2}$ be two programs in the high-level \textsc{CharonLang} representation.
Then $\mathit{Canon}(\mathit{Cert}_H(P_{1})) = \mathit{Canon}(\mathit{Cert}_H(P_{2}))$ if and only if $\mathit{Cert}_H(P_{1}) = \mathit{Cert}_H(P_{2})$.
\end{theorem}

\begin{proof}
The proof for both directions is presented below: \\

$\Leftarrow$ If $\mathit{Cert}_H(P_1) = \mathit{Cert}_H(P_2)$, then $\mathit{Canon}(\mathit{Cert}(P_1)) = \mathit{Canon}(\mathit{Cert}(P_2))$.

The proof follows directly from the determinism of the $\mathit{Canon}$ algorithm (Algorithm~\ref{algo_inv_structural}).
As the inputs are identical, then it will produce the same result in both cases. \\

$\Rightarrow$ If $\mathit{Canon}(\mathit{Cert}(P_1)) = \mathit{Canon}(\mathit{Cert}(P_2))$, then $\mathit{Cert}_H(P_1) = \mathit{Cert}_H(P_2)$.

Let $\mathit{Cert_H(P_1)} = n_1$ and $\mathit{Cert_H(P_2)} = n_2$, and
$n_1 = f_1 \times f_2 \times \ldots$ and $n_2 = f_1' \times f_2' \times \ldots$.
We have that $F = [f_1, f_2, \ldots, ]$ and $F' = [f_1', f_2', \ldots, ]$, where $F$ (and $F'$) is the list of factors that \textsc{Canon} (Algorithm~\ref{algo_inv_structural}) passes to \textsc{CanonRec} as the first argument.
The proof follows by induction on the length of $F$.
We shall demonstrate that $F = F'$ so that these two invocations of \textsc{CanonRec} produce the same program.
On the base case, we have that $F = [] = F'$.
Thus, $\textsc{CanonRec}([], \mathit{empty\ program}, [], 0) = \textsc{CanonRec}([], \mathit{empty\ program}, [], 0)$, and it follows that $\mathit{empty\ program} = \mathit{empty\ program}$.

On the inductive case, we assume that the theorem holds up to a list of factors of length $k$.
Thus, we need to show that $\textsc{CanonRec}([f_k, \ldots], P_k, S_k, v_k) = \textsc{CanonRec}([f_k', \ldots], P_k', S_k', v_k')$.
We assume, by induction, that $P_k = P_k'$, $S_k = S_k'$ and $v_k = v_k'$.
In what follows, we analyze some of the cases that Algorithm~\ref{algo_inv_structural} handles.

\begin{description}
\item [Constant:] Assume $\textsc{CanonRec}([f, \ldots], P, S, v) = \textsc{CanonRec}([f', \ldots], P, S, v)$.
If the factor $f$ is the certificate of a constant, then it has the exponent $e = 11^{c}$.
The only way the equality can hold is if the factor $f'$ has the exponent $e' = 11^{c'}$, where $c = c'$, following \textbf{else if} \textsc{IsConstant}.

\item [Variable definition:] Assume $\textsc{CanonRec}([f, \ldots], P, S, v) = \textsc{CanonRec}([f', \ldots], P, S, v)$.
If $f$ encodes the definition of a variable, it has an exponent in the form $e = 13^{\mathit{var.\ def.\ exp.}}$.
The equality will only hold if $f'$ also encodes a variable definition with an exponent $e' = 13^{\mathit{var.\ def.\ exp.'}}$, so \textsc{CanonRec} will parse $f$ and $f'$ with the case \textbf{if} \textsc{IsVarDef}.
Additionally, it is necessary that $\mathit{var.\ def.\ exp.} = \mathit{var.\ def.\ exp.'}$ for the types $\tau$ and $\tau'$ the algorithm extracts with \textsc{GetVarType} to be identical.
The equality between $\tau$ and $\tau'$ is required for \textsc{CanonRec} to add the same global variable, or struct with the same fields, or array with the same length to the canonical program being reconstructed.

\item [Variable usage:] Assume $\textsc{CanonRec}([f, \ldots], P, S, v) = \textsc{CanonRec}([f', \ldots], P, S, v)$.
If $f$ is a factor that encodes the usage of a variable, then its exponent must be $e = 17^{\mathit{var.\ usage\ exp.}}$.
For the equality to hold, $f'$ must also have a variable usage encoding exponent, $e' = 17^{\mathit{var.\ usage\ exp.'}}$, so it will fall into the same case --- \textbf{else if} \textsc{IsVarUsage}.
In particular, $f$ and $f'$ must encode the usage of the same variable, such that $\mathit{var.\ usage\ exp.} = \mathit{var.\ usage\ exp.'}$, and the same variable prime $\mathit{vp}$ can be retrieved from it.
This is mandatory for \textsc{CanonRec} to obtain the same variable name from \textsc{GetVarNameFromVarUsage} --- this method will return the name \texttt{var\_n}, where $\mathit{n}$ is the index of the variable prime $\mathit{vp}$ in the ordered set of prime numbers.

\item [Operation:] Assume $\textsc{CanonRec}([f, \ldots], P, S, v) = \textsc{CanonRec}([f', \ldots], P, S, v)$.
If $f$ is the certificate of an operation, then \textsc{CanonRec} will parse it with the \textbf{else if} \textsc{IsOperation} case.
It is required that $f$ and $f'$ encode the same operation for the equality to hold, so \textsc{GetOperationFromExp} will return the same symbol in both cases.

\end{description}

Thus, we have that if $\textsc{CanonRec}(F, \mathit{empty\ program}, [], 0) = \textsc{CanonRec}(F', \mathit{empty\ program}, [], 0)$, then $F = F'$, and therefore, $\mathit{Cert}_H(P_1) = \mathit{Cert}_H(P_2)$.
\end{proof}

The canonical form is also preserved by the compilation process.
In other words, if two programs $P_1$ and $P_2$ share the same canonical form, then they compile to the same low-level program $P_L$.
Theorem~\ref{thm_lowLevelCanon} proves this property.

\begin{theorem}[Preservation of Canonical Forms]
\label{thm_lowLevelCanon}
Let $P_1$ and $P_2$ be two programs such that $\mathit{Canon}(\mathit{Cert}_H(P_1)) = \mathit{Canon}(\mathit{Cert}_H(P_2))$.
In this case, $\mathit{Comp}(P_1) = \mathit{Comp}(P_2)$.
\end{theorem}

\begin{proof}
Theorem~\ref{thm_uniquenessCanon} demonstrates that if two programs $P_1$ and $P_2$ share the same canonical form, then they must have the same certificate.
From this property, we conclude that $\mathit{Cert}_H(P_1) = \mathit{Cert}_H(P_2)$ because $\mathit{Canon}(\mathit{Cert}_H(P_1)) = \mathit{Canon}(\mathit{Cert}_H(P_2))$, from the hypothesis.

Theorem~\ref{thm_uniquenessLowLevel}, on the other hand, shows that two programs with the same certificate will compile to the same low-level program.
As we've learned that $\mathit{Cert}_H(P_1) = \mathit{Cert}_H(P_2)$ from the previous step, then $\mathit{Comp}(P_1) = \mathit{Comp}(P_2)$.
\end{proof}

\section{Evaluation}
\label{sec_eval}

This section explores the following research questions:

\begin{enumerate}
\item \textbf{RQ1:}
How does the certificate length scale with the size of the high- and low-level source programs (AST nodes and instructions)?

\item \textbf{RQ2:}
 How does the running time of the high- and low-level certification algorithms grow with program size?
\end{enumerate}

\paragraph{Hardware/Software}
Experiments evaluated in this section were performed on an Apple M1 CPU, featuring 8 GB of integrated RAM, clock of 3.2 GHz, and banks of 192 KB and 128 KB L1 caches --- four of each.
The experimental setup runs on macOS Sequoia 15.4.1.
The \textsc{Charon} project is implemented in Python version 3.12.4 (Jun'24).

\paragraph{Benchmarks}
The GitHub repository of the \textsc{Charon} project\footnote{\url{https://github.com/guilhermeolivsilva/project-charon}} contains 20 programs written in \textsc{CharonLang}. Each program exercises different features of the high-level language. These programs vary in size, which can be measured in terms of the number of statements, AST nodes, and instructions. This diversity allows us to evaluate the asymptotic behavior of the certification technique proposed in this paper. The corresponding results are summarized in Figures~\ref{fig_certificateLength} and~\ref{fig_certificateRunningTime}, and will be discussed in the remainder of this section.

\subsection{RQ1: On the Size of Certificates}
\label{sub_eval_certificate_size}

The compilation certificate is represented as a number whose magnitude grows exponentially with the quantity of constructs in the target program, since each additional construct introduces a new multiplicative factor.  
However, the certification process does not require computing or storing the exact value of this number. Instead, the certificate is maintained in its symbolic form; that is, as a string representing the sequence of multiplicative expressions used to construct it.  
Therefore, we measure the size of a certificate based on its symbolic representation, defined as the character length of its string form, rather than by the numeric magnitude of the final value.
In this section, we show that although the absolute value of a certificate grows exponentially, its symbolic representation grows linearly with the size of the program.

\begin{figure}[ht]
\centering
\includegraphics[width=1\textwidth]{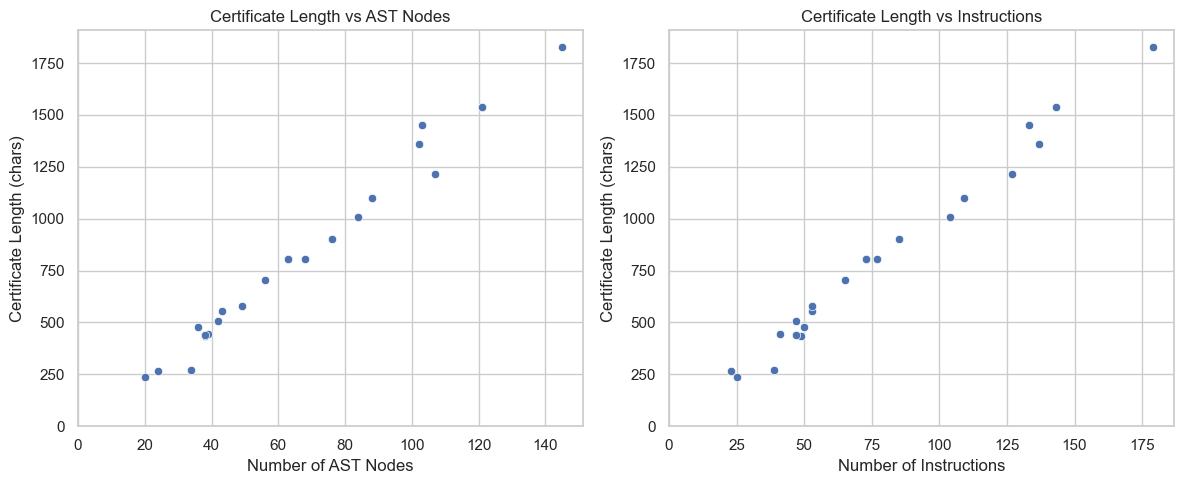}
\caption{Certificate expression length (in characters) vs program size.}
\Description{Certificate expression length (in characters) vs program size.}
\label{fig_certificateLength}
\end{figure}

\paragraph{Discussion}  
Figure~\ref{fig_certificateLength} illustrates the relationship between program size—measured by the number of AST nodes—and certificate length.  
The size of the certificates grows linearly with the size of the input program, both in terms of high-level AST nodes and low-level instructions.  

In both cases, the linear correlation is very strong: the Pearson correlation coefficient ($R^2$) between the number of AST nodes and the certificate length is 0.98. Similarly, the $R^2$ value between the length of certificates and the number of instructions in the \textsc{CharonIR} representation is 0.97.  
This linear trend reflects the design of the certification rules, where each program construct contributes a constant-size component to the overall certificate expression.

\subsection{RQ2: On the Performance of the Certifier}
\label{sub_eval_performance}

We evaluate algorithmic performance by measuring the average running time of the certification procedures across a suite of test programs, using Python's \texttt{\%\%timeit} for high-resolution timing.
Our current implementation, written in Python, is not optimized for performance and inherits the inefficiencies associated with the language.  
However, in this section, we demonstrate that the asymptotic complexity of the certification process grows linearly with the size of the input program.  
This suggests that a reimplementation in a systems-oriented language such as C, C++ or Rust, would likely result in significantly faster certification times.

\begin{figure}[ht]
\centering
\includegraphics[width=1\textwidth]{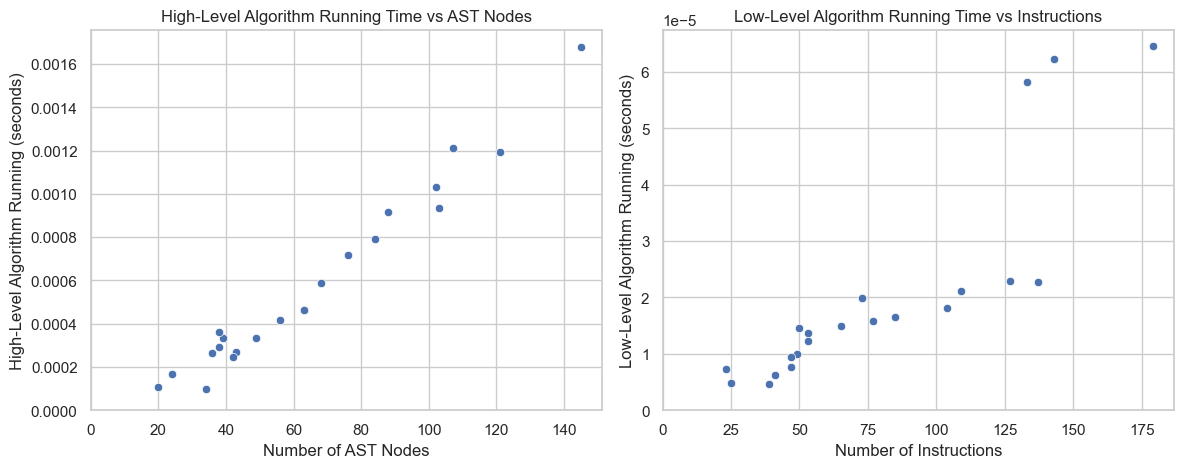}
\caption{Running time (in seconds) vs program size.}
\Description{Running time (in seconds) vs program size.}
\label{fig_certificateRunningTime}
\end{figure}

\paragraph{Discussion}
Figure~\ref{fig_certificateRunningTime} presents the correlation between the running time of the certification process and the size of the certified programs.
The left-hand plot shows the performance of the high-level algorithm (as a function of AST nodes), and the right-hand plot analyzes the low-level algorithm (as a function of instruction count).
The high-level certification algorithm runs in sub-millisecond time across the entire benchmark suite.
As shown in Figure~\ref{fig_certificateRunningTime} (left), the runtime scales linearly with the number of AST nodes, reaching approximately 1.6 ms for the largest inputs, which contains 145 nodes.
The coefficient of correlation between running time and size is 0.96; hence, very close to a perfectly linear correlation.
This behavior confirms the expected complexity of the algorithm, which performs a recursive descent over the abstract syntax tree with constant-time processing at each node.

The low-level certification algorithm, shown in Figure~\ref{fig_certificateRunningTime} (right), exhibits even faster execution times, typically in the tens of microseconds.  
It processes the instruction stream through a sequential scan based on pattern matching.
Although the low-level certification procedure demonstrates near-linear behavior in practice ($R^2 = 0.93$), it is not guaranteed to run in linear time.  
In particular, the algorithm may exhibit quadratic behavior in the presence of deeply nested branches.
This super-linear behavior is observable in our results: the three largest programs in the benchmark suite deviate from the linear trend.  
These outliers contain several control-flow constructs (e.g., \texttt{if} and \texttt{while} statements), which trigger multiple executions of the \textit{Analyze-Conditional-Jumps} pass.

\section{Related Work}
\label{sec_rw}

This paper introduces what we call ``{\it Structural Certification}'', a certification methodology in which the layout of the binary itself encodes evidence of its derivation from the source code.  
We believe this perspective is novel when compared to prior work.  
Nevertheless, as mentioned in Section~\ref{sec_introduction}, there exists a substantial body of research that shares our goal: building trust between users and compilers.  
In the remainder of this section, we review several techniques aligned with this objective.

\subsection{Diverse Double-Compiling}
\label{sub_ddc}

Diverse Double-Compiling (DDC)~\cite{Wheeler05,Wheeler10}, along with its many adaptations~\cite{Holler15,Skrimstad18,Somlo20,Drexel25}, is a defense technique proposed by \citeauthor{Wheeler05} in his PhD dissertation.  
It seeks to detect malicious behavior introduced by a compromised compiler by comparing the outputs of independently built compiler binaries.  
The method proceeds in four steps:

\begin{enumerate}
  \item Compile the source code of the compiler $C$ under test using a trusted compiler $C_t$, producing an executable compiler $X_t = C_t(C)$.
  \item Compile the same source using the suspect compiler $C_s$, producing an executable compiler $X_s = C_s(C)$.
  \item Use both resulting compiler executables to compile a test program $p$.
  \item Compare the resulting binaries, i.e., $b_t = X_t(p)$ and $b_s = X_s(p)$.
\end{enumerate}

The key insight is that a trusted compiler $C_t$, assumed to be free of malicious logic, should produce a clean version $X_t$ of the compiler $C$. 
If the executables $X_t$ and $X_s$ produce different binaries in Step 3, this discrepancy can be detected through syntactic comparison in Step 4.  

DDC serves as a practical defense against self-propagating attacks such as Thompson’s ``Trusting Trust'' backdoor.
However, it relies on two assumptions~\cite{Carnavalet14}: (i) the existence of at least one trusted compiler $C_t$, and (ii) that both $X_t$ and $X_s$ produce deterministic outputs for the same input.  
Furthermore, any differences observed in Step 4, such as $b_t \neq b_s$, require additional investigation to determine whether they stem from benign implementation differences or intentional tampering.

This paper and Diverse Double-Compiling (DDC) share the same goal: increasing confidence that a compiler is not inserting malicious code into its output.
However, they belong to different categories. DDC is a comparative technique: it detects discrepancies by building the same compiler with multiple toolchains and comparing outputs.
The paper's approach, instead, establishes a structural correspondence (a morphism) between two languages: the high-level source language and the low-level target language.
In this sense, the two techniques are complementary: the compiler described in the paper could be used as one of the compilers in a DDC setup.

\subsection{Lightweight Certification Approaches}
\label{sub_lightweight}

Lightweight certification techniques, such as Proof-Carrying Code (PCC)~\cite{Necula97} and Typed Assembly Language (TAL)~\cite{Morrisett99}, enable validation of binary code without requiring trust in the compiler.  
They do so by embedding verifiable metadata either alongside or within the compiled program.  
In PCC, the compiler generates a logical proof that the binary satisfies a predefined safety policy; this proof is distributed with the code and checked by the consumer prior to execution.  
TAL follows a similar approach, embedding type information directly in the assembly code to guarantee properties such as control-flow correctness and memory safety.  

Our work shares the same high-level goal; that is, allowing untrusted code to be verified independently; however, it adopts a different methodology.  
Rather than attaching external proofs or annotations to the code, we encode a certificate of correct compilation directly into the structure of the binary, using a G\"{o}del-style numbering schema.  
This certificate is not an auxiliary artifact.
It is instead embedded into the form of the binary code itself, enabling verification through basic arithmetic checks.  
In contrast to PCC and TAL, which require logical or type-based reasoning engines, our approach supports self-contained, solver-free certification by leveraging the regularity and reversibility of the compilation process.

\subsection{Symbolic Certification}
\label{sub_symbolic}

Symbolic certification~\cite{Barthe05,Rival04,Govereau12,Zaks08,Kirchner05} refers to techniques that verify the correctness of compiled code by constructing and comparing symbolic representations of program semantics at both the source and binary levels, ensuring that they exhibit behaviorally equivalent execution.  
While sharing a similar goal to our approach, symbolic certification is technically distinct and offers a different set of guarantees.  
We illustrate this distinction by analyzing the work of \citet{Rival04}.

In \citeauthor{Rival04}'s framework, the correctness of compilation is established through the construction of \textit{symbolic transfer functions} (STFs), which encode the semantics of source and binary code in a compositional manner.  
Correctness is then verified by demonstrating that these symbolic descriptions are behaviorally equivalent.  
This approach has two advantages over ours: (i) it supports optimized compiler backends, and (ii) it enables the verification of semantic properties such as the absence of runtime errors or the preservation of invariants.

By contrast, our technique encodes a certificate of correct compilation directly within the binary.
Instead of relying on symbolic reasoning or logical inference, we use a numeric encoding schema that maps constructs in the source program to patterns in the compiled binary.  
This yields a form of \textit{certification by structure}, which, unlike symbolic certification, requires no external proofs, solvers, or trusted verifiers.  
Although symbolic techniques like that of \citeauthor{Rival04} offer broader expressiveness and can validate aggressive compiler optimizations, they still ask for trust in the verification infrastructure.  
In contrast, our method enables a self-contained form of certification, where correctness can be checked arithmetically from the structure of the program alone.

\section{Conclusion}
\label{sec_conclusion}

This paper has introduced a new certification methodology, which we call \textit{structural certification}.
In this approach, the certificate of correct compilation is inherently embedded in the structure of the binary code itself.
Unlike previous techniques, such as proof-carrying code, translation validation, or verified compilation, which ask for trust in external artifacts, our approach enables certification through the form and layout of the compiled program.
Inspired by G\"{o}del's numbering system, our method encodes syntactic and semantic properties of the source program directly into the binary using arithmetic encodings derived from prime factorizations.
This design enables lightweight verification that requires neither theorem provers nor trusted third-party verifiers.
We have demonstrated the feasibility of this idea through the development of \textsc{Charon}, a compiler for a subset of C capable of certifying programs written in \textsc{FaCT}.
By embedding certification directly into the binary representation, our method opens a new path for certifying compilers---one in which the evidence of correct compilation is inseparable from the code it certifies.

\paragraph{Limitations}
While structural certification offers a novel and lightweight approach to trusted compilation, it comes with limitations.
Most notably, the current method assumes a non-optimizing, syntax-directed compiler, where each source construct maps deterministically to a corresponding binary pattern.
This limitation restricts support for aggressive low-level optimizations, such as instruction reordering or register allocation.
These transformations, if applied on the low-level code would invalidate the numeric correspondence required by our encoding.
Thus, programs certified via our approach must be optimized at the source-code level.
Moreover, because the certification relies on the preservation of positional and syntactic structure, it may not scale naturally to full-featured programming languages or to target architectures with highly complex instruction sets.
Finally, our method verifies structural fidelity between source and binary, but does not presently capture deeper semantic properties, such as memory safety or functional correctness.

\paragraph{Future Work}
These limitations suggest several directions for future research.
One avenue is to extend the structural encoding to tolerate limited classes of compiler optimizations, perhaps by generalizing the encoding to sets of acceptable patterns or by embedding structured metadata alongside numeric codes. 
Another possibility is to enrich the certification mechanism to express and verify semantic properties, such as safety or constant-time guarantees, while retaining the lightweight, arithmetic-checkable format.
Additionally, it would be valuable to explore how structural certification interacts with techniques like diverse double-compiling or translation validation, potentially combining their strengths to build higher-assurance compilation pipelines.
Finally, further work is needed to assess the robustness of structural certification in adversarial settings, where attackers may attempt to tamper with or forge encodings in compiled binaries.

\paragraph{Software}
The certifying infrastructure described in this paper is available at
\url{https://github.com/guilhermeolivsilva/project-charon}.

\section*{On the Usage of AI-Assisted Technologies in the Writing Process}
During the preparation of this manuscript, the authors used ChatGPT (version 4o), a language model developed by OpenAI, to revise grammar and style in portions of the text.  
This tool was used on July of 2025, and all AI-assisted content was subsequently reviewed and edited by the authors. 
The authors accept full responsibility for the content and conclusions presented in this publication.

\bibliography{references}

\end{document}